\documentclass[sigconf]{acmart}

\usepackage{bm}
\usepackage[caption=false,font=footnotesize]{subfig}  %
\usepackage{multirow,tabularx}  %
\usepackage{xcolor}
\AtBeginDocument{%
  \providecommand\BibTeX{{%
    \normalfont B\kern-0.5em{\scshape i\kern-0.25em b}\kern-0.8em\TeX}}}

\copyrightyear{2020} 
\acmYear{2020} 
\setcopyright{othergov}\acmConference[MM '20]{Proceedings of the 28th ACM International Conference on Multimedia}{October 12--16, 2020}{Seattle, WA, USA}
\acmBooktitle{Proceedings of the 28th ACM International Conference on Multimedia (MM '20), October 12--16, 2020, Seattle, WA, USA}
\acmPrice{15.00}
\acmDOI{10.1145/3394171.3413804}
\acmISBN{978-1-4503-7988-5/20/10}

\renewcommand\footnotetextcopyrightpermission[1]{} %
\acmSubmissionID{612}

\begin{document}
\fancyhead{} 

\title[Norm-in-Norm Loss]{Norm-in-Norm Loss with Faster Convergence and Better Performance for Image Quality Assessment}

\author{Dingquan Li}
\email{dingquanli@pku.edu.cn}
\orcid{0000-0002-5549-9027}
\affiliation{%
  \institution{NELVT, LMAM, School of Mathematical Sciences \& BICMR, Peking University}
  }

\author{Tingting Jiang}
\email{ttjiang@pku.edu.cn}
\affiliation{%
  \institution{NELVT, Department of Computer Science, Peking University}
  }
  
\author{Ming Jiang}
\email{ming-jiang@pku.edu.cn}
\affiliation{%
  \institution{NELVT, LMAM, School of Mathematical Sciences \& BICMR, Peking University}
  }

\renewcommand{\shortauthors}{Li et al.}

\begin{abstract}
Currently, most image quality assessment (IQA) models are supervised by the MAE or MSE loss with empirically slow convergence.
It is well-known that normalization can facilitate fast convergence.
Therefore, we explore normalization in the design of loss functions for IQA.
Specifically, we first normalize the predicted quality scores and the corresponding subjective quality scores. 
Then, the loss is defined based on the norm of the differences between these normalized values.
The resulting ``Norm-in-Norm'' loss encourages the IQA model to make linear predictions with respect to subjective quality scores.
After training, the least squares regression is applied to determine the linear mapping from the predicted quality to the subjective quality.
It is shown that the new loss is closely connected with two common IQA performance criteria (PLCC and RMSE).
Through theoretical analysis, it is proved that the embedded normalization makes the gradients of the loss function more stable and more predictable, which is conducive to the faster convergence of the IQA model.
Furthermore, to experimentally verify the effectiveness of the proposed loss, it is applied to solve a challenging problem: quality assessment of in-the-wild images.
Experiments on two relevant datasets (KonIQ-10k and CLIVE) show that, compared to MAE or MSE loss, the new loss enables the IQA model to converge about 10 times faster and the final model achieves better performance.
The proposed model also achieves state-of-the-art prediction performance on this challenging problem.
For reproducible scientific research, our code is publicly available at \url{https://github.com/lidq92/LinearityIQA}.
\end{abstract}




\maketitle

\section{Introduction}
\label{sec:introduction}

Image quality assessment (IQA) has received considerable attention~\cite{wang2004image,mittal2012no,ye2012unsupervised,kang2014convolutional,ma2016group,hosu2019koniq} and plays a key role in many vision applications, such as compression~\cite{rippel2019learned} and super-resolution~\cite{zhang2019ranksrgan}.
It can be achieved by subjective study or objective models. 
Subjective study uses mean opinion score (MOS) to assess image quality.
This is considered as the most reliable and accurate way, whereas it is expensive and time-consuming.
So the objective models that can automatically predict image quality are in urgent need.
In terms of the availability of the reference image, objective IQA models can be divided into three categories: full-reference IQA~\cite{wang2004image,zhang2014vsi,kim2017deep0,bosse2018deep}, reduced-reference IQA~\cite{ma2011reduced,xu2015fractal,bampis2017speed,liu2018reduced}, and no-reference IQA~\cite{mittal2012no,kang2014convolutional,liu2017rankiqa,ren2018ran4iqa}.

Most classic learning-based IQA models are based on mapping the handcrafted features to image quality by support vector regression (SVR)~\cite{mittal2012no}.
Recently, deep learning-based models, which jointly learn feature representation and quality prediction, show great promise in IQA~\cite{kang2014convolutional,ren2018ran4iqa,lin2018hallucinated,hosu2019koniq}.
However, these models mostly treat IQA as a general regression problem.
And they adopt standard regression loss functions for training, \textit{i.e.}, mean absolute error (MAE) and mean square error (MSE) between the predicted quality scores and the corresponding subjective quality scores.
We notice a fact that the IQA models trained using MAE or MSE loss exhibit slow convergence. 
For example, on a dataset containing only about 10,000 images with a resolution of 664$\times$498, training the model on an NVIDIA GeForce RTX 2080 Ti GPU (11GB) takes more than one day to reach convergence.
Since the size of the training dataset becomes larger and larger in the deep learning era, faster convergence is preferable to reduce the training time.

In this work, we tackle the slow convergence problem in the context of IQA. 
In fact, slow convergence problem is common in machine learning and computer vision, which may be due to the non-smooth loss landscape~\cite{santurkar2018does}.
To provide faster convergence for the learning process, it is widely-used to do input data normalization or intermediate feature rescaling~\cite{ioffe2015batch,hoffer2018norm}.
Normalizing the output predictions is rarely recommended.
However, it is shown that normalizing the network output can lead to faster convergence of generative networks for image super-resolution~\cite{mullery2018batch}.
Inspired by this work, to achieve fast convergence of IQA model training, we explore normalization in the design of loss functions for IQA. 

\begin{figure*}[!htb]
    \centering
    \includegraphics[width=.85\linewidth]{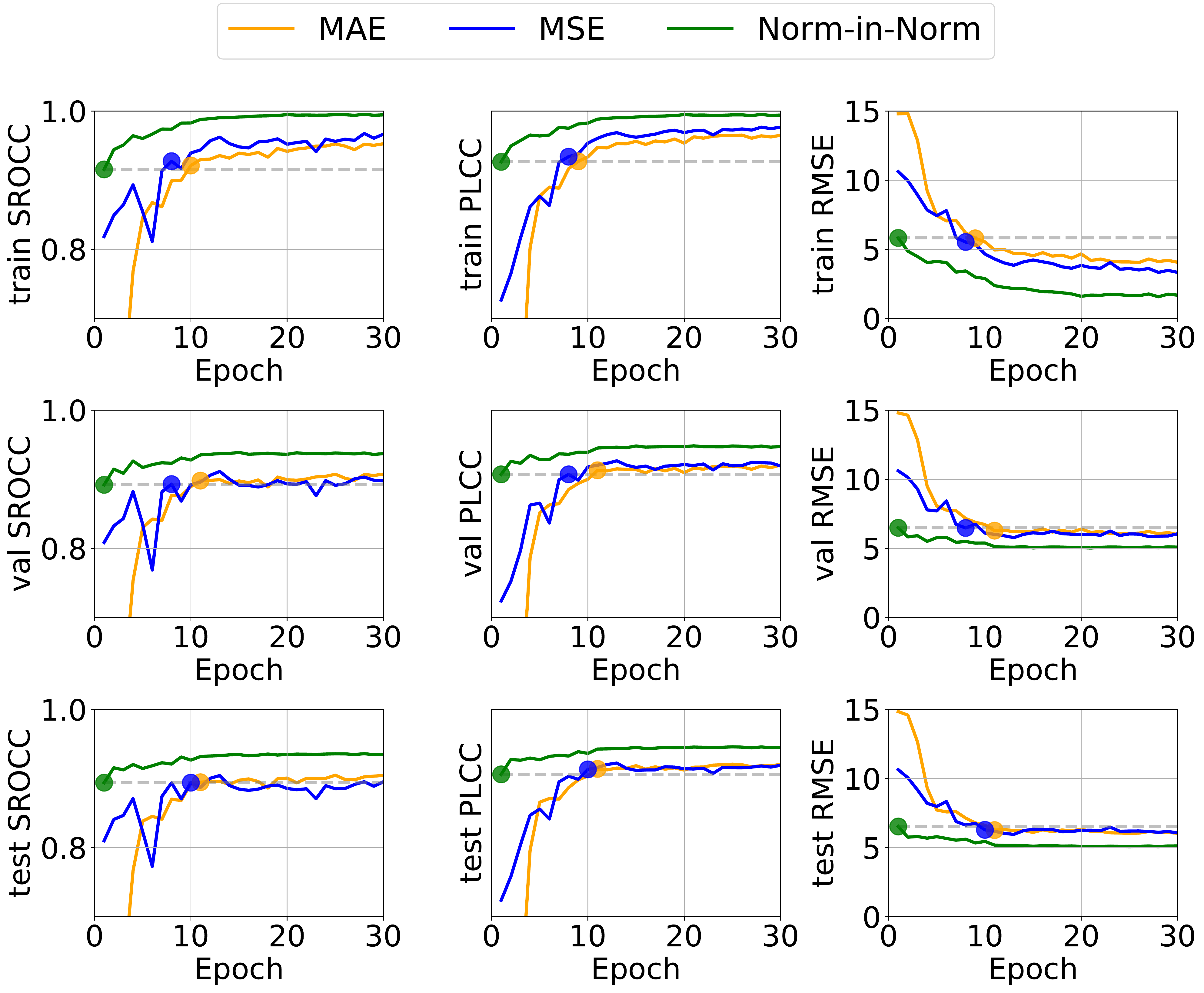}
    \caption{The training/validation/testing curves on KonIQ-10k of the models trained with MAE, MSE, and the proposed ``Norm-in-Norm'' loss. SROCC, PLCC and RMSE are three criteria for benchmarking IQA models, where larger SROCC/PLCC and smaller RMSE indicate better prediction performance. The circle marker shows the first time it surpasses a prediction performance indicated by the grey dash line.}
    \label{fig:loss}
\end{figure*}

We propose a class of loss functions, denoted as ``Norm-in-Norm'', for training an IQA model with fast convergence.
Specifically, the predicted quality scores is firstly subtracted by their mean, and then they are divided by their norm after centralization. 
Similar normalization is applied to the subjective quality scores.
After the normalization, we define the new loss based on the norm of the differences between the normalized values.
The new loss normalizes both labels and predictions while label normalization only normalizes labels, and the new loss encourages the IQA model to make linear predictions with respect to (w.r.t.) subjective quality scores.
Hence, after training, the linear relationship can be determined by applying least squares regression (LSR) on the whole training set for mapping the predictions to the subjective quality scores. 
In the testing phase, this learned linear relationship is applied to the model prediction to get the final predicted quality score for a test image.

There are two interesting findings about the new loss.
First, we derive that Pearson's linear correlation coefficient (PLCC)-induced loss~\cite{ma2018geometric,liu2018end} is a special case of the proposed loss, where PLCC is a criterion for benchmarking IQA models.
Second, after introducing a variant of the proposed loss, we show its connection to root mean square error (RMSE) --- another IQA performance criterion.

Further, we conduct theoretical analysis on the property of the new loss. 
And it is proved that due to the embedded normalization, the new loss has stronger Lipschitzness and $\beta$-smoothness~\cite{nesterov2013introductory}, which means the gradients of the loss function is more stable and more predictable. 
Thus, the gradient-based algorithm for learning the IQA model has a smoother loss landscape.
And this is conducive to the faster convergence of the IQA model.

Generally, the proposed ``Norm-in-Norm'' loss can be applied to any regression problem, including IQA problems.
In particular, we pick a challenging real-world IQA problem: quality assessment of in-the-wild images, to verify the effectiveness of the proposed loss.
Quality assessment of in-the-wild images has two challenges, \textit{i.e.}, content dependency\footnote{Content dependency means human ratings of image quality depend on image content} and distortion complexity.
We design a no-reference IQA model based on aggregating and fusing multiple level deep features extracted from the image classification network.
Deep features are considered for tackling the content dependency. 
And multiple level features are used for handling the distortion complexity.
The model is trained with the proposed loss. 
The experiments are conducted on two benchmark datasets, \textit{i.e.}, CLIVE~\cite{ghadiyaram2016massive} and KonIQ-10k~\cite{hosu2019koniq}.
Figure~\ref{fig:loss} shows the convergence results on KonIQ-10k.
By using the proposed loss, the model only needs to look at the training images once (\textit{i.e.}, in one epoch) to achieve a prediction performance indicated by the grey dash line.
It is about 10 times faster than MAE loss and MSE loss.
This verifies that the proposed loss facilitates faster convergence for training IQA models.
We claim that it mainly benefits from the embedded normalization in our loss.
Besides the faster convergence, we notice that the proposed loss also leads to a better prediction performance than MAE loss and MSE loss.

To sum up, our main contribution is that we propose a class of normalization-embedded loss functions in the context of IQA.
The new loss is shown to have some connections to PLCC and RMSE.
And our theoretical analysis proves that the embedded normalization in the proposed loss can facilitate faster convergence.
For the quality assessment of in-the-wild images, it is experimentally verified that the new loss can provide both better prediction performance and faster convergence than MAE loss and MSE loss. 
What's more, the proposed model outperforms state-of-the-art models on this challenging IQA problem.

\section{``Norm-in-Norm'' Loss}
\label{sec:loss}

In this section, we explore the normalization in the design of loss functions, and propose a class of ``Norm-in-Norm'' loss functions for training IQA models with fast convergence.
The idea is to apply normalization for the predicted quality scores and the subjective quality scores respectively using their own statistics when computing the loss. 

Assume we have $N$ images on the training batch. 
For the $i$-th image $I_i$, we denote the predicted quality by an objective IQA model $F(\cdot; \bm{\theta})$ as $\hat{Q}_{i}$ and its subjective quality score (\textit{i.e.}, MOS) as $Q_{i}$.

\subsection{Loss Computation}
\label{sec:loss computation}
Our loss computation can be generally divided into three steps: computation of the statistics, normalization based on the statistics, and loss as the norm of the differences between the normalized values.
The left part of Figure~\ref{fig:norm-in-norm loss} shows an illustration of the forward path of the proposed loss. 
We detail each step in the following.

\paragraph{Computation of the Statistics} 
First, given the predicted quality scores $\hat{\mathbf{Q}}=(\hat{Q}_{1}, \cdots, \hat{Q}_{N})$, we calculate their mean $\hat{a}$. 
Similarly, given the subjective quality scores ${\mathbf{Q}}=({Q}_{1}, \cdots, {Q}_{N})$, their mean $a$ is calculated. 
The $L^q$-norm of the centered values is then computed, respectively.
\begin{align}
\hat{a} &= \frac{1}{N}\sum_{i=1}^{N} \hat{Q}_{i},\ 
\hat{b} = {\left(\sum_{i=1}^{N} {|\hat{Q}_{i}-\hat{a}|}^{q}\right)}^{\frac{1}{q}}, \\
a &= \frac{1}{N}\sum_{i=1}^{N} {Q}_{i},\ 
b = {\left(\sum_{i=1}^{N} {|{Q}_{i}-{a}|}^{q}\right)}^{\frac{1}{q}}\label{eq:stats2},
\end{align}
where $q\ge 1$ is a hyper-parameter. $\hat{b}$ and $b$ are the norm of the centered predicted quality scores and the norm of the centered subjective quality scores, respectively.

\paragraph{Normalization Based on the Statistics}
Second, we normalize the predicted quality scores and the subjective quality scores based on their own mean and centered norm, respectively. That is, we first subtract the mean from the predicted/subjective quality scores and then divide them by the norm. 
\begin{align}
\hat{S}_i &= \frac{\hat{Q}_i-\hat{a}}{\hat{b}}, i=1,\cdots,N, \label{eq:normalization1}\\
{S}_i &= \frac{{Q}_i-{a}}{{b}}, i=1,\cdots,N,\label{eq:normalization2} 
\end{align}
where $\hat{\mathbf{S}}=(\hat{S}_1,\cdots,\hat{S}_N)$ are the normalized predicted quality scores, and ${\mathbf{S}}= ({S}_1,\cdots,{S}_N)$ are the normalized subjective quality scores.

\paragraph{Loss As the Norm of the Differences}
The final step computes the differences $\hat{\mathbf{S}}-{\mathbf{S}}$ between the normalized predicted quality scores $\hat{\mathbf{S}}$ and the normalized subjective quality scores ${\mathbf{S}}$.
Then the loss $l$ is defined as the $p$-th power of the $L^p$-norm of the differences ($p\ge 1$), and it is normalized to $[0, 1]$.
\begin{equation}
    l(\hat{\mathbf{Q}}, \mathbf{Q}) = \frac{1}{c}\sum_{i=1}^{N} {|\hat{S}_{i}-{S}_{i}|}^{p}, 
\end{equation}
where $c$ is a normalization factor. $c$ can be determined using Minkowski inequality and H\"older's inequality (see the \textbf{Supplementary Materials}~\ref{sec:supp_a}), and it follows the following equation.

\begin{equation}\label{eq:factor}
c = 
\begin{cases}
    2^pN^{1-\frac{p}{q}} & \mbox{if}\quad p<q, \\
    2^p & \mbox{if}\quad p\ge q.
\end{cases}
\end{equation}

Based on the forward propagation, we can easily conduct the backward propagation by the chain rule, which is described in the right part of Figure~\ref{fig:norm-in-norm loss}.

\begin{figure}[!t]
    \centering
    \includegraphics[width=\linewidth]{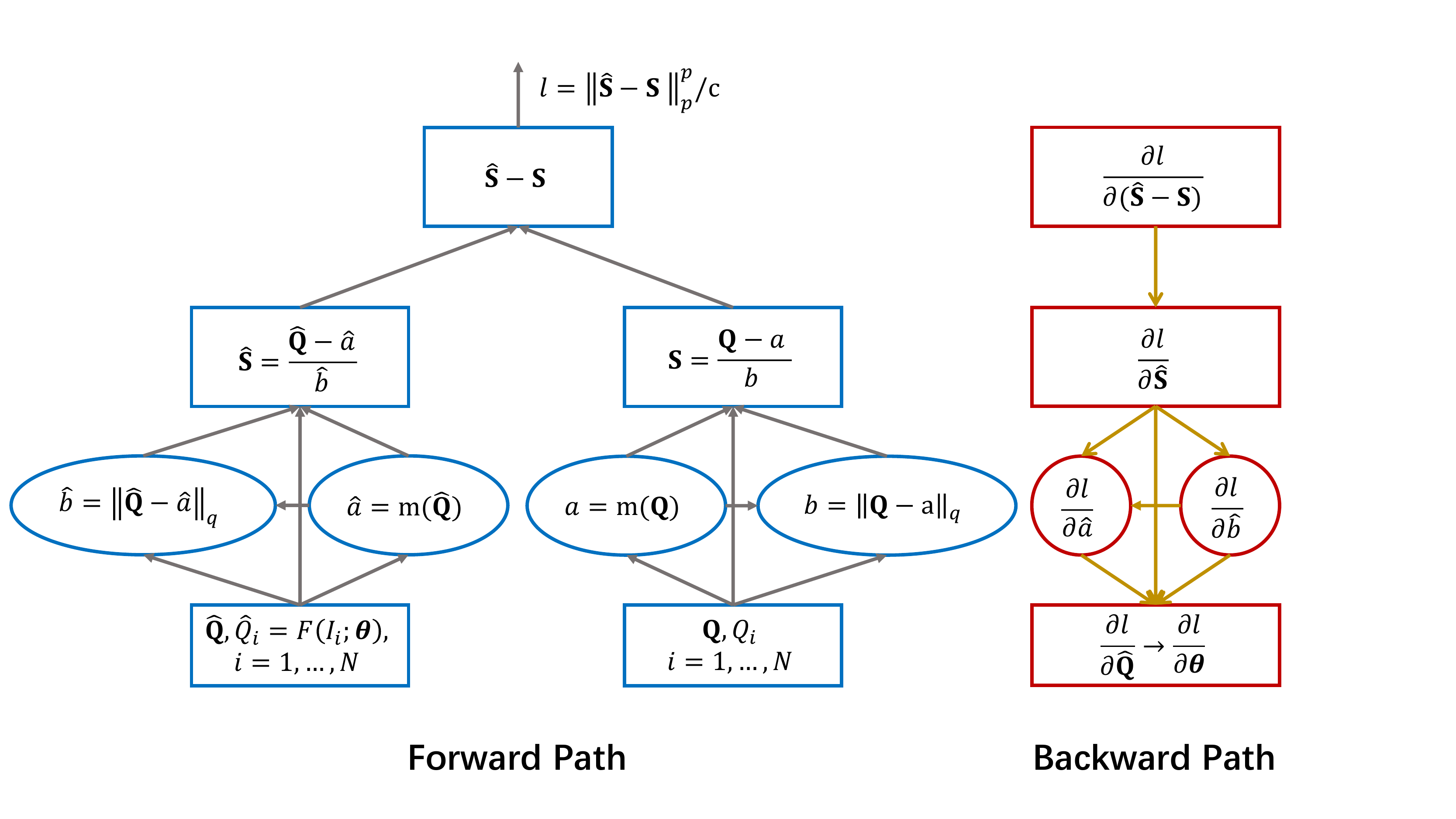}
    \caption{Illustration of the forward and backward paths of the proposed loss. $\mathrm{m}(\cdot)$ denotes the mean function. $F(\cdot; \bm{\theta})$ is the IQA model, where $\bm{\theta}$ represents the model parameters.}
    \label{fig:norm-in-norm loss}
\end{figure}

Specifically, based on Eqn.~(\ref{eq:normalization1}), we have
\begin{equation}\label{eq:grad1}
    \frac{\partial \hat{S}_i}{\partial \hat{Q}_j} = \frac{1}{\hat{b}}\left\{1_{i=j}-\frac{1}{N}-\hat{S}_i\frac{|\hat{S}_j|^q}{\hat{S}_j}+\hat{S}_i\frac{1}{N}\sum_{k=1}^{N}\frac{|\hat{S}_k|^q}{\hat{S}_k}\right\}.
\end{equation}
where $1_{i=j}$ is an indicate function, and it equals 1 if $i=j$, otherwise, 0.

Denote
\begin{equation}\label{eq:rmu}
    \hat{R}_j = \frac{|\hat{S}_j|^q}{\hat{S}_j}, \mu_{\hat{R}} = \frac{1}{N}\sum_{k=1}^{N}\frac{|\hat{S}_k|^q}{\hat{S}_k} = \frac{1}{N}\langle\mathbf{1}, \hat{\mathbf{R}}\rangle.
\end{equation}

Then
\begin{align}\label{eq:grad}
\frac{\partial l}{\partial \hat{Q}_j} = \sum_{i=1}^{N} \frac{\partial l}{\partial \hat{S}_i} \frac{\partial  \hat{S}_i}{\partial \hat{Q}_j} 
= \frac{1}{\hat{b}}\sum_{i=1}^{N} \frac{\partial l}{\partial \hat{S}_i} \left\{1_{i=j}-\frac{1}{N}-\hat{S}_i\hat{R}_j+\hat{S}_i\mu_{\hat{R}}\right\}.
\end{align}

So we get the derivative of $l$ w.r.t. $\hat{\mathbf{Q}}$, \textit{i.e.}, $\frac{\partial l}{\partial \hat{\mathbf{Q}}}$.
\begin{align}
\frac{\partial l}{\partial \hat{\mathbf{Q}}} = \frac{1}{\hat{b}} \left\{\frac{\partial l}{\partial \hat{\mathbf{S}}}-\frac{1}{N}\mathbf{1}^T\left\langle\mathbf{1}, \frac{\partial l}{\partial \hat{\mathbf{S}}}\right\rangle-\hat{\mathbf{R}}^T\left\langle\frac{\partial l}{\partial \hat{\mathbf{S}}},\hat{\mathbf{S}}\right\rangle+\mu_{\hat{R}}\mathbf{1}^T\left\langle\frac{\partial l}{\partial \hat{\mathbf{S}}},\hat{\mathbf{S}}\right\rangle\right\},
\end{align}
which includes four terms. 
The first term is related to $\frac{\partial l}{\partial \hat{\mathbf{S}}}$.
The second term is related to $\frac{\partial l}{\partial \hat{a}}$.
The thrid term is related to $\frac{\partial l}{\partial \hat{b}}$.
And the fourth term is related to both $\frac{\partial l}{\partial \hat{a}}$ and $\frac{\partial l}{\partial \hat{b}}$.

\noindent\textbf{Remark}: The normalization in Eqn.~(\ref{eq:normalization1}) is invariant to linear predictions. 
That is, for any $k_1\hat{Q}+k_2 (k_1^2+k_2^2\neq 0)$, we derive the same $\hat{S}$.
So, we have
\begin{equation}
    l(k_1\hat{\mathbf{Q}}+k_2, \mathbf{Q}) = l(\hat{\mathbf{Q}}, \mathbf{Q}).
\end{equation}
The loss encourages the IQA model to make predictions that are linearly correlated with the subjective quality scores.

\subsection{A Special Case: PLCC-induced Loss}
Pearson's Linear Correlation Coefficient (PLCC), $\rho$, is a criterion for benchmarking IQA models, which is defined as follows.
\begin{equation}\label{eq:plcc}
\rho(\hat{\mathbf{Q}}, {\mathbf{Q}}) =\frac{\sum_{i=1}^N(\hat{Q}_{i}-\hat{a})(Q_{i}-a)}{\sqrt{\sum_{i=1}^N{(\hat{Q}_{i}-\hat{a})}^2\sum_{i=1}^N{(Q_{i}-a)}^2}}.
\end{equation}

In the following, we will prove that PLCC-induced loss $(1-\rho)/2$ is a special case of the ``Norm-in-Norm'' loss functions.

First, when $q$ is set to 2 in the ``Norm-in-Norm'' loss, $\hat{Q_i}\rightarrow\hat{S_i}$ and ${Q_i}\rightarrow{S_i}$ defined in Eqn.~(\ref{eq:normalization1}-\ref{eq:normalization2}) relate to the well-known z-score transformation. 
And $\hat{S_i}$ and ${S_i}$ have the following properties.
\begin{equation}\label{eq:unit}
    \sum_{i=1}^N \hat{S}_{i}^{2}=\sum_{i=1}^N {S}_{i}^{2}=1.
\end{equation}

With the notation of $\hat{S_i}$ and $S_i$, we can reformulate PLCC.
\begin{align}\label{eq:aloss}
\rho(\hat{\mathbf{Q}}, {\mathbf{Q}}) &=\frac{\sum_{i=1}^N(\hat{Q}_{i}-\hat{a})(Q_{i}-a)}{\sqrt{\sum_{i=1}^N{(\hat{Q}_{i}-\hat{a})}^2\sum_{i=1}^N{(Q_{i}-a)}^2}} \nonumber\\
& =\sum_{i=1}^N\frac{\hat{Q}_{i}-\hat{a}}{\hat{b}}\frac{Q_{i}-a}{b}=\sum_{i=1}^N\hat{S}_{i}S_i\nonumber\\
&=1-\frac{1}{2}\sum_{i=1}^N(\hat{S}_{i}-S_i)^2.\ \mbox{[using Eqn.~(\ref{eq:unit})]} \nonumber
\end{align}

When $p$ equals to 2, we have $c=4$ using Eqn.~(\ref{eq:factor}), then we derive the following equation.

\begin{equation}\label{eq:relationship1}
    \rho = 1 - \frac{cl}{2}=1-2l\Longleftrightarrow l = (1-\rho)/2. 
\end{equation}
That is, PLCC-induced loss is equivalent to the ``Norm-in-Norm'' loss $l(\hat{\mathbf{Q}}, \mathbf{Q})$ when the $p, q$ are all set to 2.

\subsection{A Variant and Its Connection to RMSE}
In this subsection, we introduce a variant of the ``Norm-in-Norm'' loss and show its connection to root mean square error (RMSE).
In the end of Section~\ref{sec:loss computation}, we remark that the ``Norm-in-Norm'' loss focuses on training an IQA model to make linear predictions w.r.t. subjective quality scores.
Under this linearity assumption, we can only require the normalized predicted quality scores $\hat{\mathbf{S}}$ and normalized subjective quality scores ${\mathbf{S}}$ to be linearly correlated. 
That is $|\rho(\hat{\mathbf{S}}, {\mathbf{S}})|=1$. 
With this expectation, we can get a variant of the ``Norm-in-Norm'' loss as follows.
\begin{equation}
    l'(\hat{\mathbf{Q}}, \mathbf{Q}) = \frac{1}{c}\sum_{i=1}^{N} {|\rho(\hat{\mathbf{S}}, {\mathbf{S}})\hat{S}_{i}-{S}_{i}|}^{p}.
\end{equation}

\paragraph{Connection to RMSE} 
We apply the least squares regression (LSR) to find the linear mapping between $\hat{Q}_i$ and ${Q}_i$.
\begin{equation}
    Q_i = k_1\hat{Q}_i+k_2, (i=1,\cdots, N),
\end{equation}
where $k_1$ and $k_2$ are two free parameters. 

It is equivalent to solving the following minimization problem.
\begin{align}
     &\min_{k_1,k_2}\sum_{i=1}^N(k_1\hat{Q}_i+k_2-Q_i)^2 \label{eq:LSR} \\
    \Longrightarrow\quad& k_1^*=\frac{\sum_{i=1}^N(\hat{Q}_{i}-\hat{a})(Q_{i}-a)}{\sum_{i=1}^N{(\hat{Q}_{i}-\hat{a})}^2}, \nonumber\\
    & k_2^*=a-k_1^*\hat{a}. \nonumber
\end{align}

We substitute $k_1^*,k_2^*$ into formula~(\ref{eq:LSR}), and can easily get the minimum loss as follows.
\begin{align}
    &\sum_{i=1}^N(k_1^*\hat{Q}_i+k_2^*-Q_i)^2\nonumber\\
    =&\sum_{i=1}^N b^2 |\rho(\hat{\mathbf{S}}, {\mathbf{S}})\hat{S}_{i}-{S}_{i}|^2\nonumber\\
    =& 4 b^2 l',\nonumber
\end{align}
where $p=2, q=2$ are considered in our loss variant $l'$, and $b$ is the centered norm of $\mathbf{Q}$ as defined in Eqn.~(\ref{eq:stats2}). 

Thus, we derive the RMSE between the linearly mapped scores and the subjective quality scores as follows.
\begin{equation}
    \mathrm{RMSE}(k_1^*\hat{\mathbf{Q}}+k_2^*, \mathbf{Q}) = \sqrt{4b^2l'(\hat{\mathbf{Q}}, \mathbf{Q})/N}.
\end{equation}
That is, a special case of the loss variant $l'(\hat{\mathbf{Q}}, \mathbf{Q})$ with $p=2, q=2$ is connected with RMSE --- another criterion for benchmarking IQA models.

\section{Theoretical Analysis}
We introduce the concepts of Lipschitzness and $\beta$-smoothness~\cite{nesterov2013introductory}. 
For a univariate function $f$, $f$ is $L$-Lipschitz if $|f(x_1)-f(x_2)| \le L|x_1-x_2|, \forall x_1, x_2$. 
And $f$ is $\beta$-smooth if its gradient is $\beta$-Lipschitz, \textit{i.e.}, $\beta$-smoothness corresponds to the Lipschitzness of the gradient. 
The proposed loss $l$ is a differentiate multivariate function w.r.t. the model predictions $\hat{\mathbf{Q}}$.
Its Lipschitzness is indicated by its gradient magnitude and its $\beta$-smoothness in the gradient direction is indicated by the quadratic form of its Hessian matrix. 
Smaller gradient magnitude and quadratic form of its Hessian correspond to better Lipschitzness and $\beta$-smoothness, respectively.

In this section, we theoretically prove that, when $q$ equals 2, the embedded normalization can improve the Lipschitzness and the $\beta$-smoothness of the proposed loss $l$. That is, the gradient magnitude and the quadratic form of its Hessian are reduced by the embedded normalization, which indicates the gradients of the proposed loss is more stable and more predictable.
This ensures that the gradient-based algorithm has a smoother loss landscape, so the training of the IQA model gets more robust and the model converges faster.

First, we prove a theorem about Lipschitzness. 
When $q=2$, based on Eqn.~(\ref{eq:rmu}), $\hat{\mathbf{R}}=\hat{\mathbf{S}}, \mu_{\hat{R}}=\frac{1}{N}\langle\mathbf{1}, \hat{\mathbf{S}}\rangle=0$. 
Denote ${\mathbf{g}}=\frac{\partial l}{\partial \hat{\mathbf{S}}}, {\mathbf{g}}_n=\frac{\partial l}{\partial \hat{\mathbf{Q}}}$. 
Thus Eqn.~(\ref{eq:grad}) becomes
\begin{align}
{\mathbf{g}}_n = \frac{1}{\hat{b}} \left\{{\mathbf{g}}-\frac{1}{N}\mathbf{1}^T\left\langle\mathbf{1}, {\mathbf{g}}\right\rangle-\hat{\mathbf{S}}^T\left\langle{\mathbf{g}},\hat{\mathbf{S}}\right\rangle\right\}.
\end{align}

We show that the gradient magnitude of the new loss satisfies Eqn.~(\ref{eq:L}) in Theorem~\ref{theorem1} (See proof in \textbf{Supplemental Materials}~\ref{sec:supp_c}). 

\begin{theorem}[Lipschitzness]\label{theorem1}
When $q$ equals 2, the gradient magnitude of the proposed loss $l$ has the following property.
\begin{align}\label{eq:L}
\left\|{\mathbf{g}}_n\right\|^2 = \frac{1}{\hat{b}^2} \left\{\left\|{\mathbf{g}}\right\|^2-\frac{1}{N}\left\langle\mathbf{1}, {\mathbf{g}}\right\rangle^2-\left\langle{\mathbf{g}},\hat{\mathbf{S}}\right\rangle^2\right\}, 
\end{align}
where the right side contains three terms. 
The first term is directly related to the gradient of the loss w.r.t. the normalized predicted quality scores, \textit{i.e.}, $\left\|{\mathbf{g}}\right\|^2$. 
The second term (non-positive) is contributed by $\frac{\partial l}{\partial \hat{a}}$. 
The third term (non-positive) is contributed by $\frac{\partial l}{\partial \hat{b}}$. 
The contributions of $\frac{\partial l}{\partial \hat{a}}$ and $\frac{\partial l}{\partial \hat{b}}$ are independent.
\end{theorem}

The left side of Eqn.~(\ref{eq:L}), $\left\|{\mathbf{g}}_n\right\|^2$, indicates the Lipschitzness with the embedded normalization. 
Without the normalization, the Lipschitzness is indicated by $\left\|{\mathbf{g}}\right\|^2$. 
From Eqn.~(\ref{eq:L}), we derive that the Lipschitzness of the proposed loss is improved whenever the sum of the gradient ${\mathbf{g}}$ deviates from 0 or the gradient ${\mathbf{g}}$ correlates the normalized predicted quality scores $\hat{\mathbf{S}}$. 
In addition, ${\hat{b}}$ is larger than 1 in practice (see \textbf{Supplemental Materials}~\ref{sec:supp_b}), which also contributes to the improvement of the Lipschitzness. 
So from Theorem~\ref{theorem1}, we can infer that the embedded normalization improves the Lipschitzness.

Next, we prove a theorem about $\beta$-smoothness. 
Denote $\mathbf{H}=\frac{\partial^2 l}{\partial \hat{\mathbf{S}}^2}, \mathbf{H}_n=\frac{\partial^2 l}{\partial \hat{\mathbf{Q}}^2}$. 
We then prove that the quadratic form of the loss Hessian in the gradient direction satisfies Eqn.~(\ref{eq:beta}) in Theorem~\ref{theorem2} (The proof is provided in the \textbf{Supplemental Materials}~\ref{sec:supp_d}). 

\begin{theorem}[$\beta$-smoothness]\label{theorem2}
When $q$ equals 2, the Hessian matrix of the proposed loss $l$ has the following property.
\begin{align}
    {\mathbf{g}}_n^T\mathbf{H}_n{\mathbf{g}}_n = \frac{1}{\hat{b}^2} \left\{{\mathbf{g}}_n^T\mathbf{H}{\mathbf{g}}_n-\langle{\mathbf{g}},\hat{\mathbf{S}}\rangle\left\|{\mathbf{g}}_n\right\|^2\right\}.
\end{align}

Further, when $p$ equals 2, we have ${\mathbf{g}} =\frac{2}{c}(\hat{\mathbf{S}}-{\mathbf{S}}), \mathbf{H}=\frac{2}{c}\mathbf{I}$, where $\mathbf{I}$ is the identity matrix of order $N$. 
The above equation becomes
\begin{align}\label{eq:beta}
    {\mathbf{g}}_n^T\mathbf{H}_n{\mathbf{g}}_n = & \frac{1}{\hat{b}^4} \left\{\mathbf{g}^T\mathbf{H}\mathbf{g}\right. \nonumber\\
    & \left.-\frac{2}{c}\left(1-\langle{\mathbf{S}},\hat{\mathbf{S}}\rangle\right)\left[\frac{4}{c^2}\left(1-\langle{\mathbf{S}},\hat{\mathbf{S}}\rangle\right)+\left\|{\mathbf{g}}\right\|^2-\left\langle{\mathbf{g}},\hat{\mathbf{S}}\right\rangle^2\right]\right\}.
\end{align}

\end{theorem}

The left side of Eqn.~(\ref{eq:beta}), ${\mathbf{g}}_n^T\mathbf{H}_n{\mathbf{g}}_n$, indicates the $\beta$-smoothness with the embedded normalization. 
Without the normalization, the $\beta$-smoothness is indicated by ${\mathbf{g}}^T\mathbf{H}{\mathbf{g}}$. 
From Eqn.~(\ref{eq:beta}), we can see that the $\beta$-smoothness is improved when the normalized subjective quality scores ${\mathbf{S}}$ and normalized predicted quality scores $\hat{\mathbf{S}}$ are not linearly correlated ($\langle{\mathbf{S}},\hat{\mathbf{S}}\rangle<1$). 
And it is further improved if the gradient ${\mathbf{g}}$ and the normalized predicted quality scores $\hat{\mathbf{S}}$ are also not linearly correlated ($\langle{\mathbf{g}},\hat{\mathbf{S}}\rangle<\|{\mathbf{g}}\|$). 
In addition, ${\hat{b}}>1$ also contributes to the improvement of the $\beta$-smoothness. 
So from Theorem~\ref{theorem2}, we can infer that the embedded normalization improves the $\beta$-smoothness.

\section{Verification on Quality Assessment of In-the-Wild Images}
\label{sec:verification}
Besides theoretical analysis, we also conduct an experimental verification.
Quality assessment of in-the-wild images is important for many real-world applications, but few attention is paid to it.
In-the-wild images contain lots of unique contents and complex distortions.
The greatest challenge for this problem is how to handle the content dependency and distortion complexity.
In this section, we pick this challenging problem for verifying the effectiveness of our ``Norm-in-Norm'' loss in comparison with MAE loss and MSE loss. 
In addition, we intend to provide state-of-the-art prediction performance on this challenging problem.

\begin{figure}[!t]
    \centering
    \includegraphics[width=.898\linewidth]{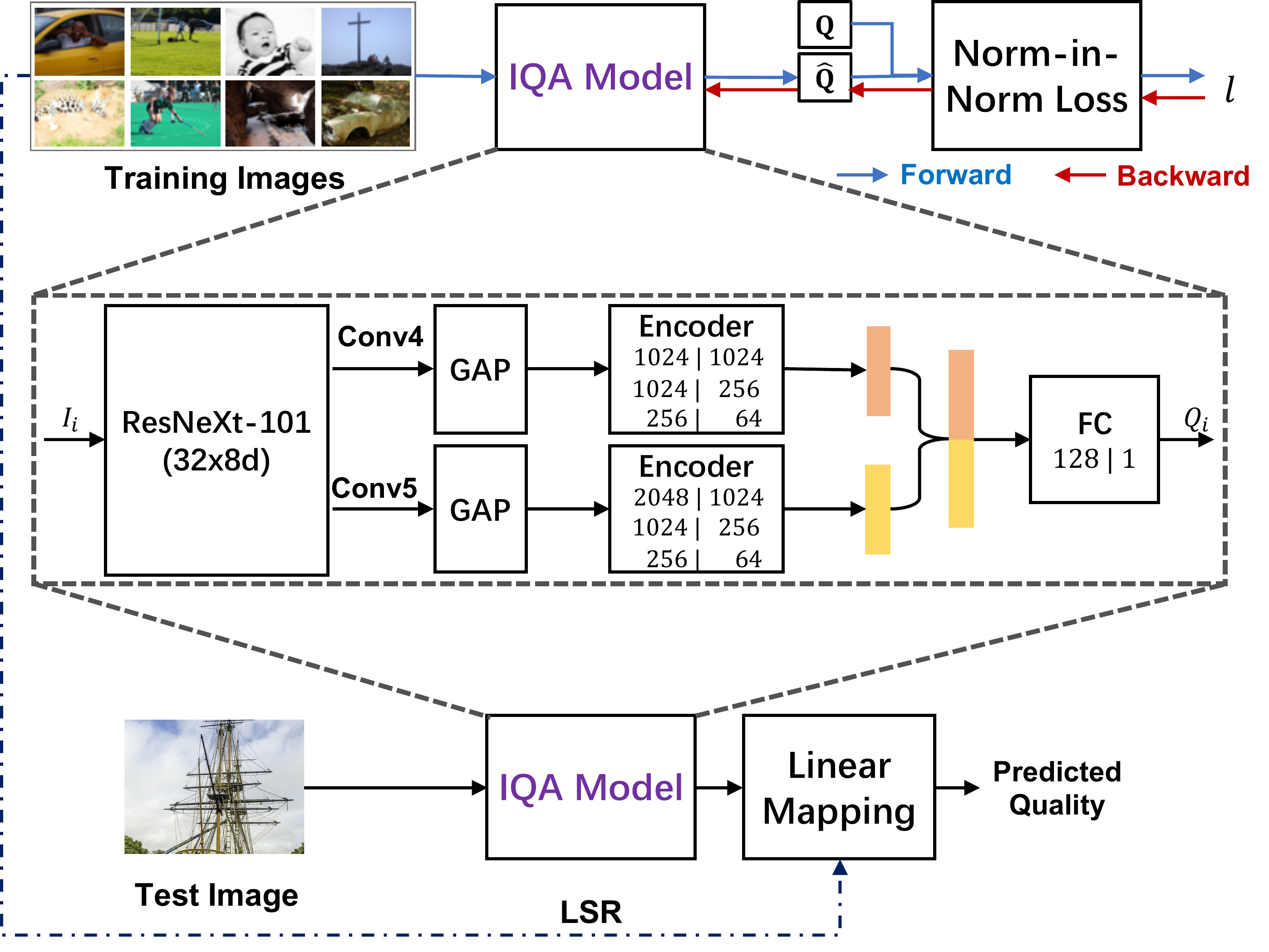}
    \caption{The framework for quality assessment of in-the-wild images. Note that $C_1 |C_2$ denotes a fully-connected (FC) layer. BN and ReLU are omitted in the illustration.}
    \label{fig:framework}
\end{figure}

\subsection{IQA Framework}
\label{sec:framework}

Our IQA framework for in-the-wild images is shown in Figure~\ref{fig:framework}.
We introduce a model that extracts deep pre-trained features for tackling content dependency and fuses multi-level features for handling the distortion complexity. 
Specifically, it first extracts multi-level deep feature extraction from an image classification backbone (\textit{e.g.}, 32x8d ResNeXt-101~\cite{xie2017aggregated}).
Then feature aggregation is achieved by global average pooling (GAP).
Next, features are encoded by an encoder with three fully-connected (FC) layers, where each FC layer is followed by a batch normalization (BN)~\cite{ioffe2015batch} layer and a ReLU activation function. 
After that, the IQA model concatenates the encoded features from different levels, and the concatenated features are finally mapped to the output by an FC layer. 
After training the IQA model, to determine the linear mapping from the predictions to the subjective quality scores, the least squares regression (LSR) method is applied on the whole training set. 
In the testing phase, based on the learned linear mapping, the prediction of the IQA model is linearly mapped to produce the final predicted quality score for a test image.

\subsection{Experimental Setup}
\label{sec:experiments}
We conduct experiments on two benchmark datasets: CLIVE~\cite{ghadiyaram2016massive} and KonIQ-10k~\cite{hosu2019koniq}.
We follow the same experimental setup as described in~\cite{hosu2019koniq}.
KonIQ-10k contains 10073 images, and they are divided into three sets: a training set (7058 images), a validation set (1000 images), and a test set (2015 images).
We train our model on the training set of KonIQ-10k, save the best performed model on the validation set of KonIQ-10k in terms of Spearman's Rank-Order Correlation Coefficient (SROCC).
We report the SROCC, PLCC, and RMSE values on the test set of KonIQ-10k for prediction performance evaluation.
CLIVE includes 1162 images, and it is used for cross-dataset evaluation.

\paragraph{Implementation Details}
The input images is resized $664\times 498$.
The backbone models for multi-level feature extraction are chosen from ResNet-18, ResNet-34, ResNet-50~\cite{he2016deep}, and ResNeXt-101~\cite{xie2017aggregated} pre-trained on ImageNet~\cite{deng2009imagenet}. 
And the features are extracted from the stage ``conv4'' and stage ``conv5'' of the backbone.
To explicitly encode the features at each level to a task-aware feature space, we add auxiliary supervision to the encoded feature at each level. 
That is, the encoded feature is directly followed by a single FC layer to output the image quality score. 
Thus, beside the main stream loss, we get another two streams of losses.
The final training loss is a weighted average of the three loss values, where the weight hyper-parameters for the main stream loss and the other two streams of losses set to 1, 0.1, 0.1, respectively. 
We train the model with NVIDIA GeForce RTX 2080 Ti GPU using Adam optimizer for 30 epochs, where the learning rate drops to its 1/10 every 10 epochs.
The initial learning rate is chosen from 1e-3, 1e-4, and 1e-5.
The batch size varies from 4, 8, 16. 
And the ratio between the learning rate of the backbone's parameters and of the other parameters, denoted as ``fine-tuned rate'', is selected from 0, 0.01, 0.1, and 1.
By default, we use an initial learning rate 1e-4, batch size 8, and fine-tuned rate 0.1.
The default values for hyper-parameters $p$ and $q$ in the ``Norm-in-Norm'' loss are 1 and 2, respectively.
The proposed model is implemented with PyTorch~\cite{paszke2019pytorch}.
To support reproducible scientific research, we have released our code at \url{https://github.com/lidq92/LinearityIQA}.

\subsection{Results and Analysis}
In this subsection, we show the experimental results and verify the proposed loss in different aspects. 

\subsubsection{Model Convergence With MAE, MSE, or the Proposed Loss}
\label{sec:loss-exp}
In this experiment, $p,q$ in the proposed loss are set to $1,2$, and we adopt the backbone ResNeXt-101 for models trained with all losses. 
The training/validation/testing curves on KonIQ-10k are shown in Figure~\ref{fig:loss}. 
Looking at the circle markers, to reach the prediction performance indicated by the grey dash line, MAE and MSE are empirically about ten times slower than the proposed loss.  
For MAE or MSE loss, to achieve a comparable prediction performance with the proposed loss, the models need to be trained with much more time.
And the final state of the convergence also indicates that our proposed loss achieves better prediction performance (higher SROCC/PLCC and lower RMSE) than MAE loss and MSE loss.
We experimentally conclude that the model trained with our proposed loss converges faster and better than that with MAE or MSE loss.

\begin{figure}[!htb]
    \centering
    \includegraphics[width=.75\linewidth]{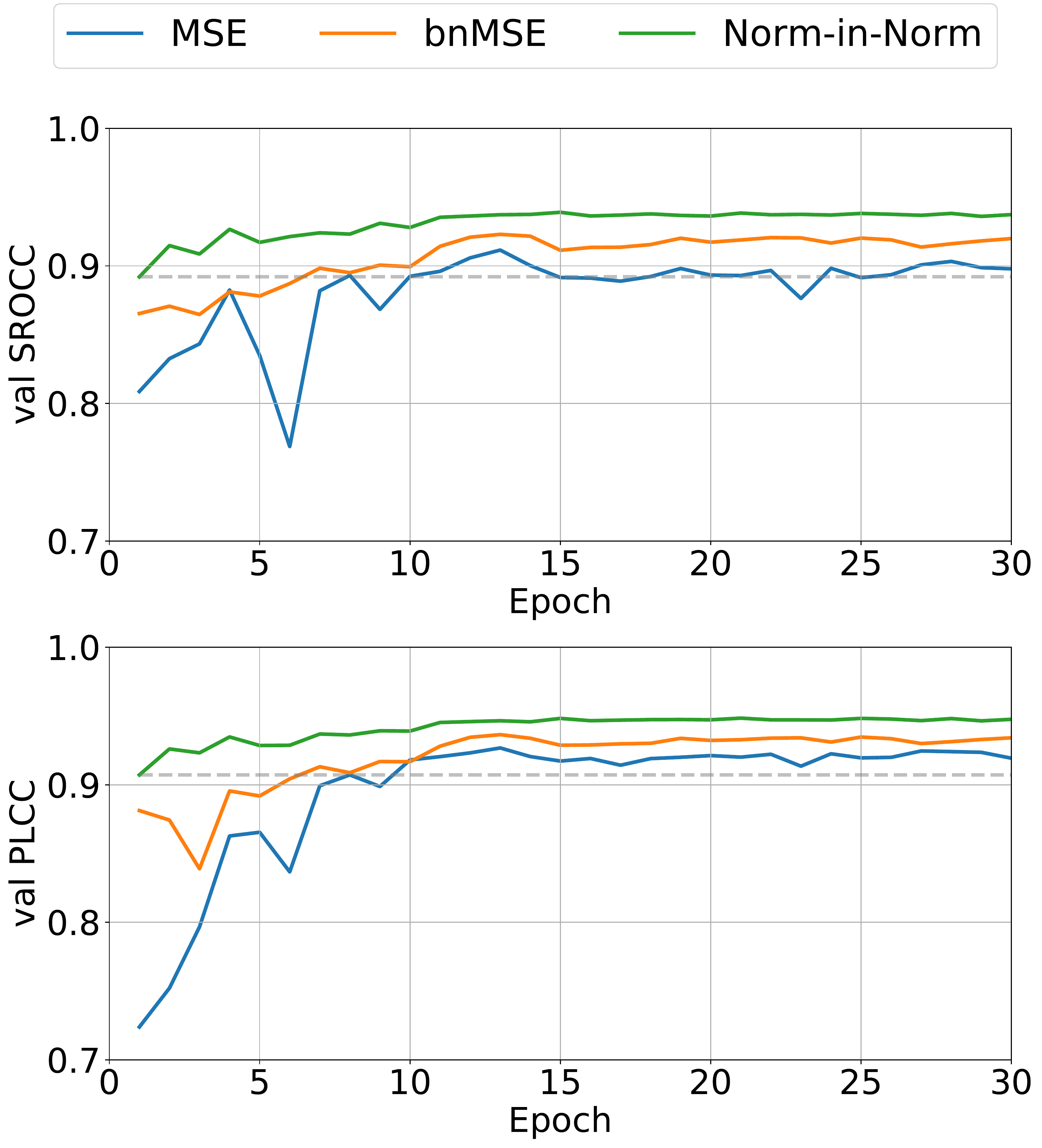}
    \caption{The validation curves on KonIQ-10k}
    \label{fig:bnMSE}
\end{figure}

Our method may look similar to adding a BN layer to the output of the current model. 
Thus, our method is also compared with ``bnMSE'', where a BN layer is added to the output of the model and the model is trained with MSE loss. 
Figure~\ref{fig:bnMSE}~shows the validation curves on KonIQ-10k. 
When compared to MSE, ``bnMSE'' leads to faster convergence and better performance. 
However, it is worse than the proposed ``Norm-in-Norm''. 
This is because the learned linear relationship in ``bnMSE'' is based on the cumulation of the batch-sample statistics, which is not accurate at the beginning.
Thus, it will slow down the convergence and somehow disturb the learning process.
On the contrast, the proposed method separates the network learning and the learning of the linear relationship, where the network first focuses on making linear predictions and then LSR is applied on the whole training set to determine a more accurate linear relationship.
Besides, it should be noted that, unlike ``bnMSE'', the proposed method does not change the architecture and it also normalizes subjective quality scores.

\begin{figure*}[!htb]
    \centering
    \subfloat[Different initial learning rates]{\includegraphics[width=.68\linewidth]{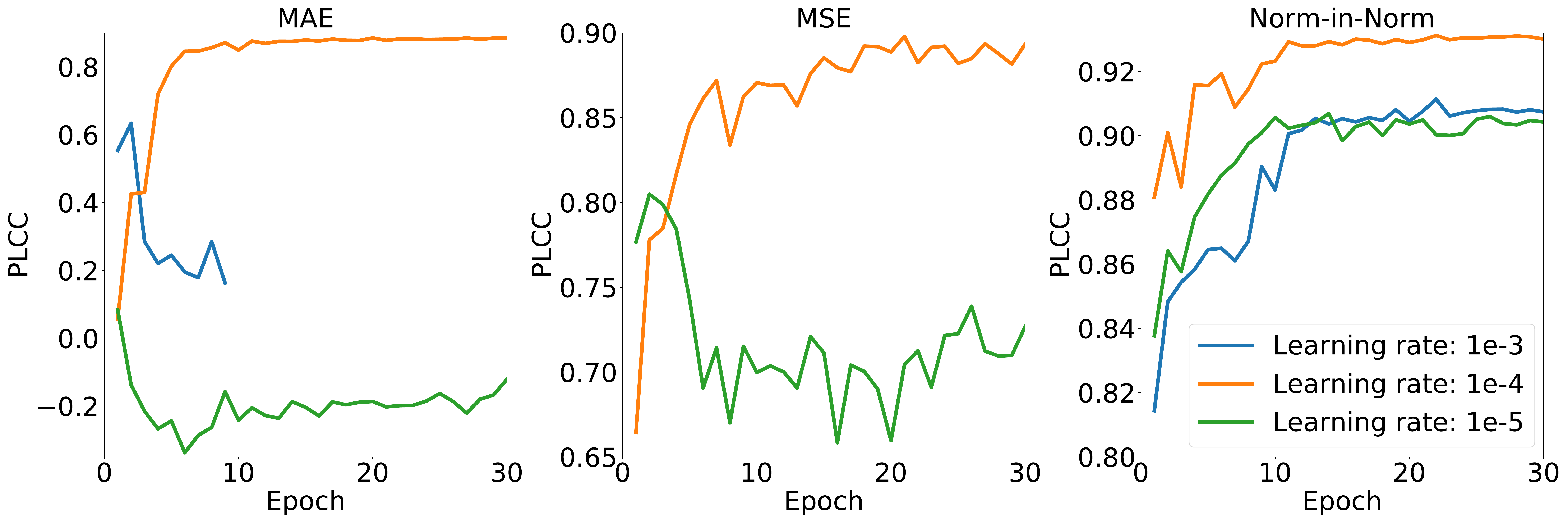}\label{fig:learningrate}}
    \hfill
    \subfloat[Different batch sizes]{\includegraphics[width=.68\linewidth]{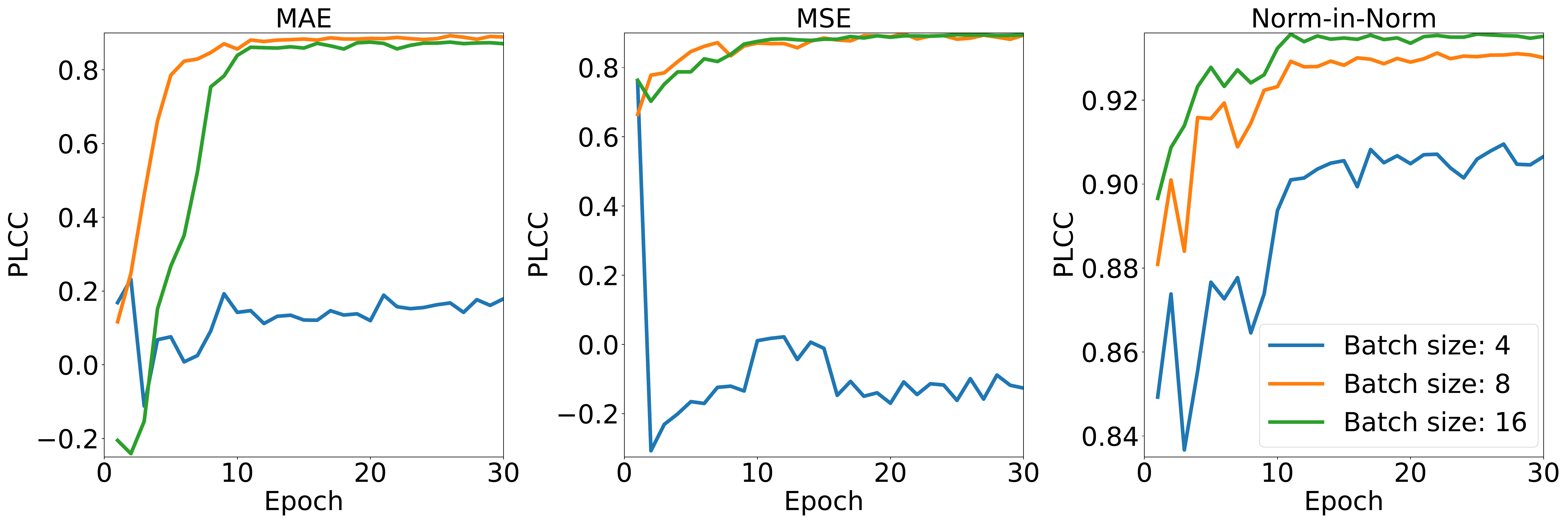}\label{fig:batchsize}}
    \hfill
    \subfloat[Different fine-tuned rates]{\includegraphics[width=.68\linewidth]{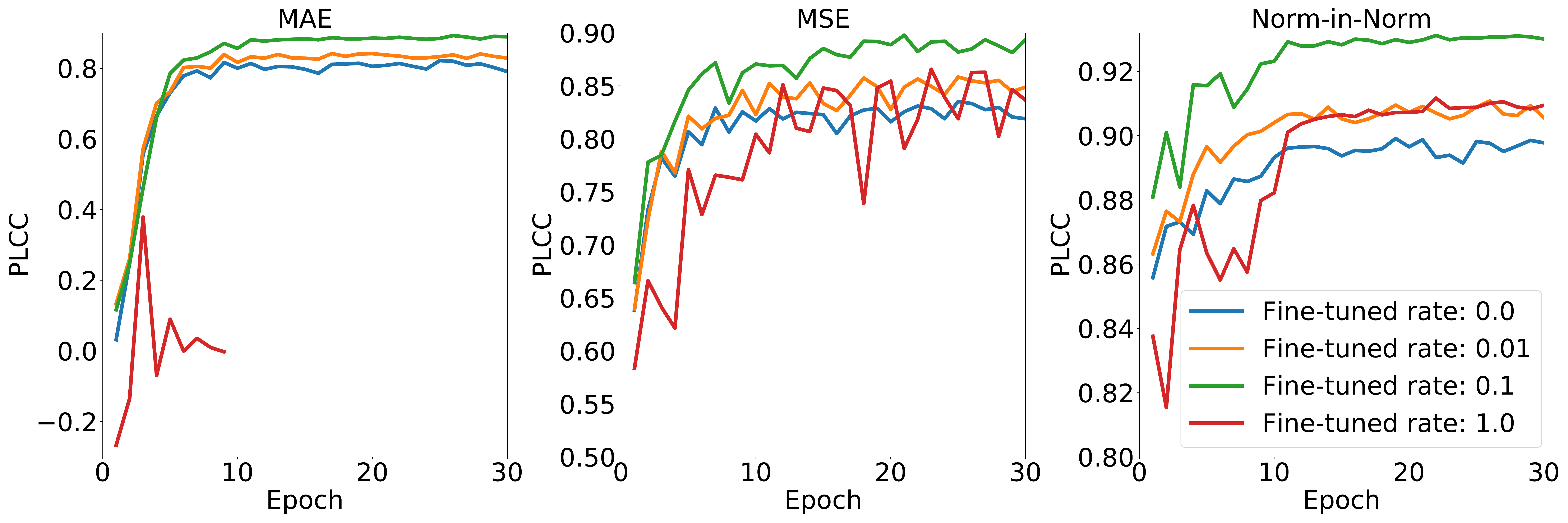}\label{fig:finetunedrate}}
    \caption{PLCC curves on KonIQ-10k validation set using models trained with MAE, MSE, or ``Norm-in-Norm'' loss. The incomplete curves indicate that the training process is stopped due to the encounter of NaNs or Infs.}
    \label{fig:sensitivity}
\end{figure*}

\begin{table}[!htb]
\centering
\caption{PLCC on KonIQ-10k test set under different $p,q$}
\label{tab:pq}
\begin{tabular}{lcccc}  
\toprule
\multirow{2}{*}{Backbone}  & $p=1$ & $p=1$ & $p=2$ & $p=2,q=2$, \textit{i.e.} \\ 
& $q=1$ & $q=2$ & $q=1$ & PLCC-induced loss \\
\midrule
ResNet-18   & 0.916 & \textbf{0.918} & 0.912 & 0.913 \\
ResNet-34   & 0.924 & \textbf{0.926} & 0.923 & 0.924 \\
ResNet-50   & 0.928 & \textbf{0.930} & 0.927 & 0.929 \\
ResNeXt-101 & 0.946 & \textbf{0.947} & 0.944 & 0.945 \\
\bottomrule
\end{tabular}
\end{table}

\subsubsection{Effects of $p, q$ in the ``Norm-in-Norm'' Loss} 
In this experiment, we explore the effect of $p,q$ in the proposed loss functions.
We consider four choices: $p=1,q=1$, $p=1,q=2$, $p=2,q=1$, and $p=2,q=2$ (\textit{i.e.}, PLCC-induced loss).
The PLCC values on KonIQ-10k test set for different $p,q$ under different backbones are shown in Table~\ref{tab:pq}.
We can see that $p=1$ is generally better than $p=2$.
This can be explained by the fact that the loss with $p=2$ is more sensitive to the outliers than the loss with $p=1$.
Although the loss with $p=1,q=1$ is slightly inferior to the loss with $p=1,q=2$, $L^1$ normalization (\textit{i.e.}, $q=1$) may improve numerical stability in low-precision implementations as pointed out in~\cite{hoffer2018norm}.
Besides, we now only focus on the task of quality assessment for in-the-wild images, and the results shows that $p=1, q=2$ is the best choice for KonIQ-10k.
However, the best $p,q$ may be task-dependent and dataset-dependent.
It deserves a further study on how to adaptively determine these hyper-parameters in a probability framework, just like the study described in~\cite{barron2019general}.
\subsubsection{Training Stability With MAE, MSE, or the Proposed Loss Under Different Learning Hyper-parameters} 
In this experiment, we consider use ResNet-50 as the backbone model, and vary the default initial learning rate, batch size, and fine-tuned rate to see the training stability under these hyper-parameters.
The validation curves on KonIQ-10k are shown in Figure~\ref{fig:sensitivity}. 
We can see that training models with MAE or MSE loss are unstable when varying the learning rates, batch sizes.
And training the model with MAE loss is unstable when fine-tuned rate is too large.
Compared to MAE loss and MSE loss, the ``Norm-in-Norm'' loss is more stable under different choices.
For all losses, the best results are achieved under an initial learning rate 1e-4 and a fine-tuned rate 0.1.
However, our ``Norm-in-Norm'' loss can achieve a better validation performance under batch size 16 than under batch size 8, while the MAE or MSE loss does not.
This is because our loss is a batch-correlated loss, and a larger batch size may lead to a more accurate estimation of the sample statistics.
For fair comparison, in other experiments, we also use the default batch size, \textit{i.e.}, 8, for the proposed loss.

\begin{figure}[!htb]
    \centering
    \subfloat[PLCC on KonIQ-10k test set]{\includegraphics[width=.75\linewidth]{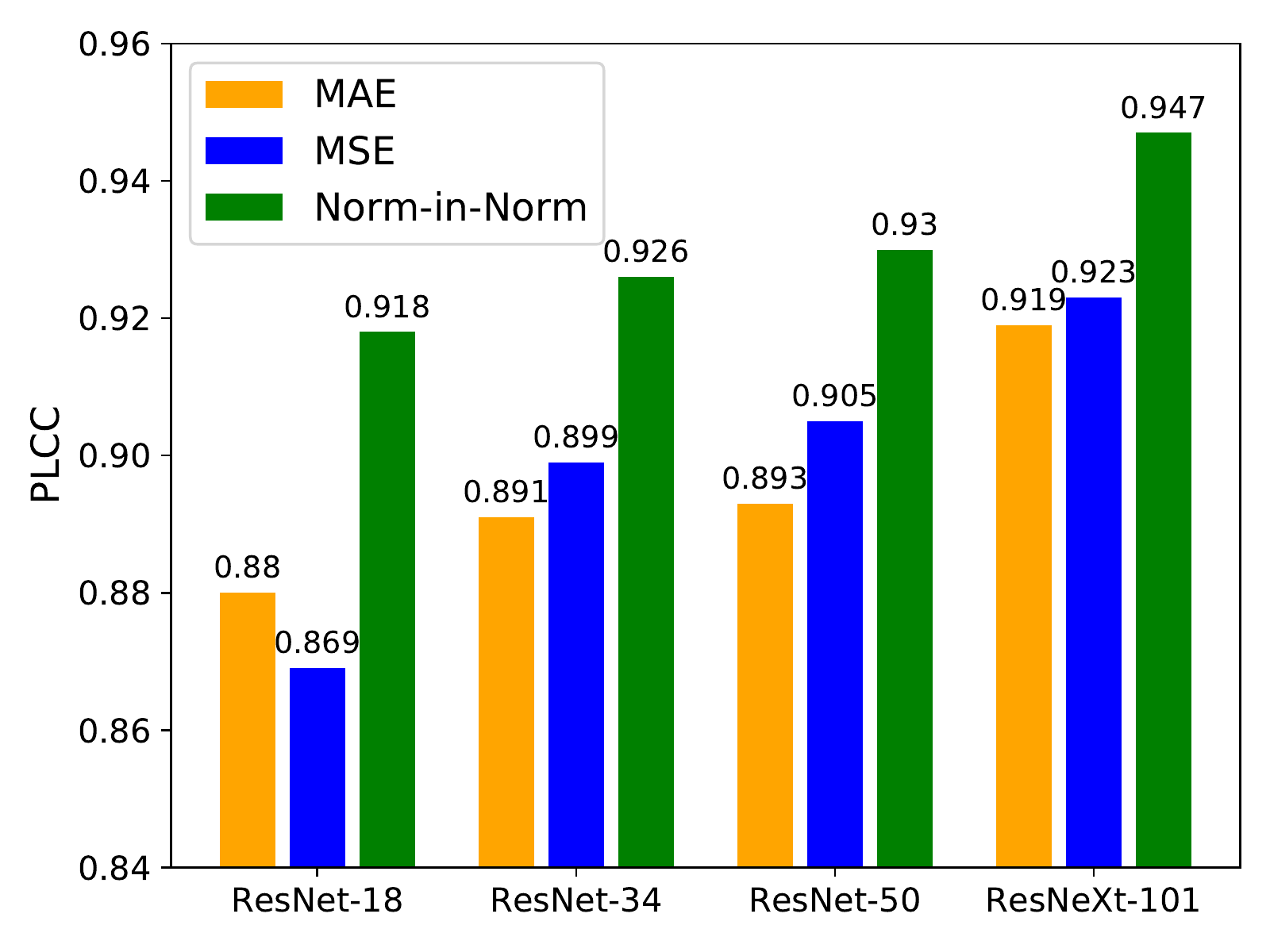}\label{fig:backbone on KonIQ-10k}}
    \hfill
    \subfloat[SROCC on CLIVE]{\includegraphics[width=.75\linewidth]{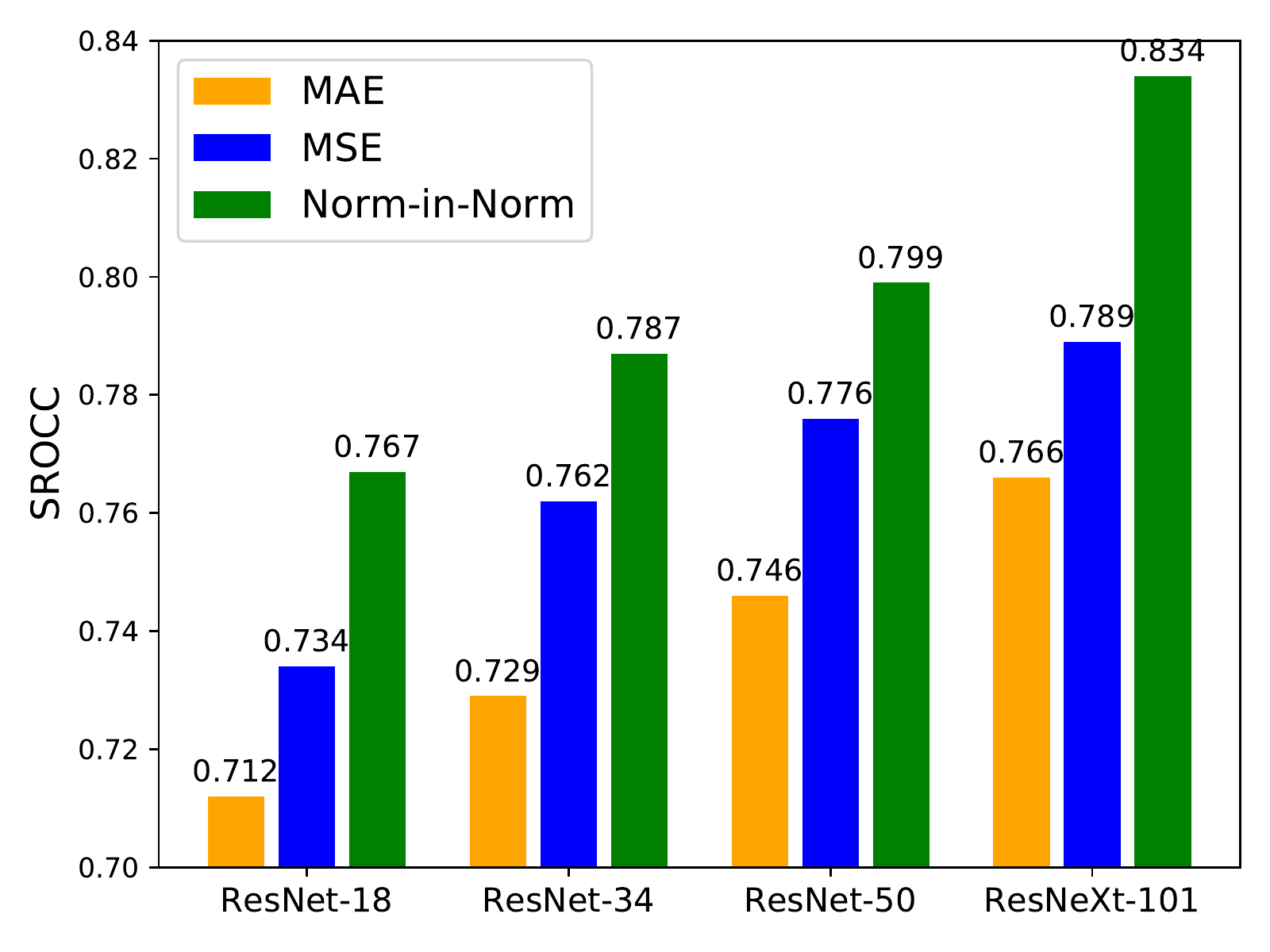}\label{fig:backbone on CLIVE}}
    \caption{Performance under different backbone models}
    \label{fig:backbone}
\end{figure}

\subsubsection{Performance Consistency for MAE, MSE, or the Proposed Loss Among Different Backbone Architectures} 
\label{sec:4.3.4}
In this experiment, we consider ResNet-18, ResNet-34, ResNet-50, and ResNeXt-101 as the backbone architectures, and train the models with MAE, MSE, and the proposed loss. 
The PLCC on KonIQ-10k test set and the SROCC on CLIVE are shown in Figure~\ref{fig:backbone}.
It can be seen that our proposed loss consistently achieves the best prediction performance under different backbone architectures.
The scatter plots between the predicted quality scores by the models using backbone ResNeXt-101 and MOSs on KonIQ-10k test set are shown in Figure~\ref{fig:scatters}.
The scatter points of the model with our loss are more centered in the diagonal line, which means a better prediction of image quality.

\begin{figure}[!htb]
    \centering
    \includegraphics[width=.9\linewidth]{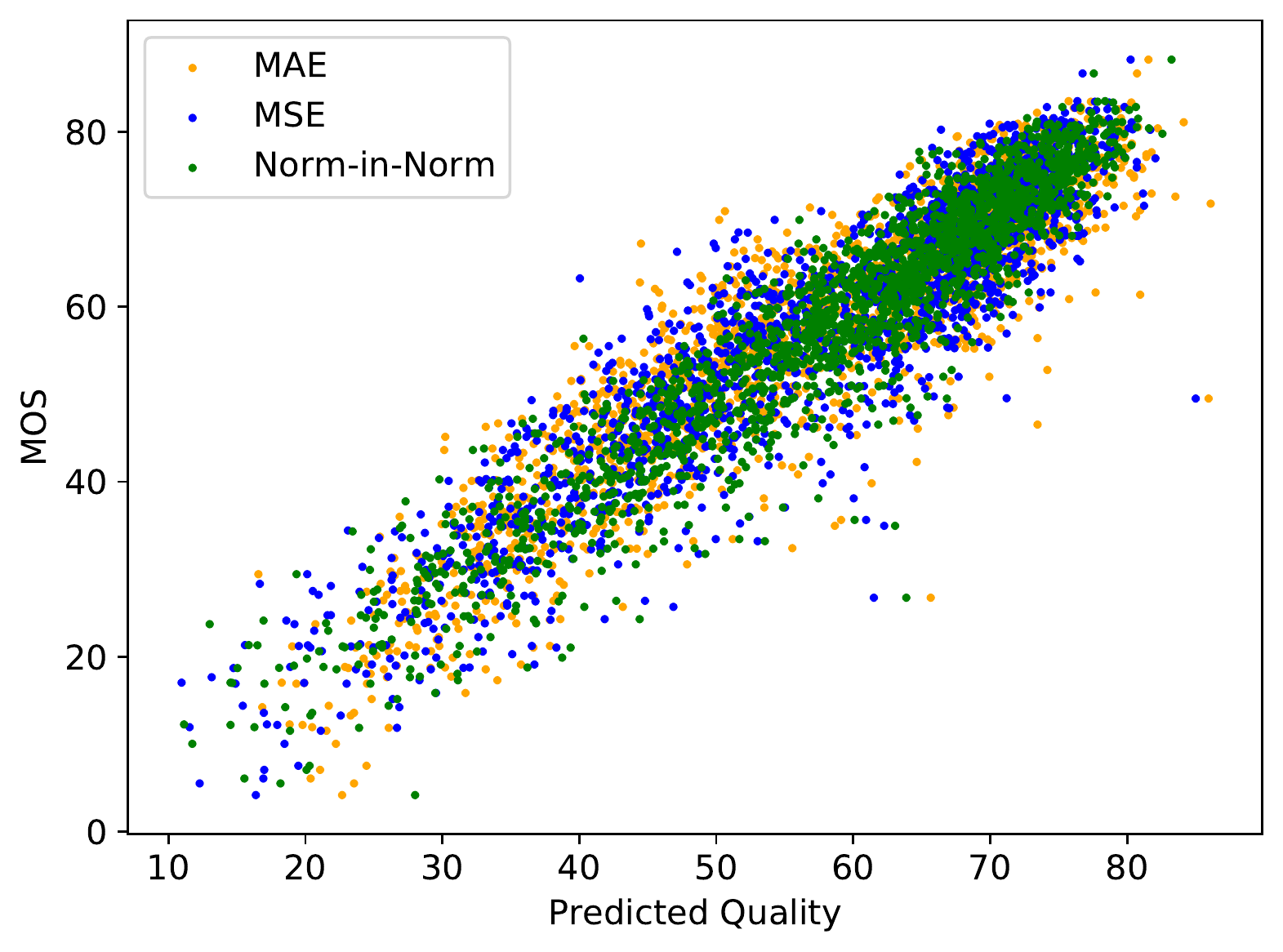}
    \caption{Scatter plots between the predicted quality and the MOS on KonIQ-10k test set}
    \label{fig:scatters}
\end{figure}

\begin{table}[!htb]
\centering
\caption{Performance comparison with SOTA on KonIQ-10k test set and the whole CLIVE}
\label{tab:SOTA}
\resizebox{\linewidth}{!}{
\begin{tabular}{llcccc}  
\toprule
\multirow{2}{*}{Model}  & \multirow{2}{*}{Loss} & \multicolumn{2}{c}{KonIQ-10k} &  \multicolumn{2}{c}{CLIVE} \\
& & SROCC & PLCC & SROCC & PLCC \\
\midrule
BRISQUE (TIP'12)   & SVR loss & 0.705  & 0.707 & 0.561  & 0.598      \\
CORNIA (CVPR'12)   & SVR loss & 0.780  & 0.808 & 0.621  & 0.644      \\
HOSA (TIP'16)      & SVR loss & 0.805  & 0.828 & 0.628  & 0.668      \\
DeepBIQ (SIViP'18) & SVR loss & 0.872 & 0.886   & 0.742 & 0.747      \\
\midrule
CNNIQA (CVPR'14)   & MAE loss & 0.572  & 0.584 & 0.465  & 0.450      \\
DeepRN (ICME'18) & Huber loss & 0.867  & 0.880 & 0.720  & 0.750      \\
KonCept512 (TIP'20) & MSE loss & 0.921 & 0.937   & 0.825 & 0.848      \\
\midrule
\multirow{2}{*}{Proposed} & $l$ & \textbf{0.937} & \textbf{0.947} & \textbf{0.834} & \textbf{0.849} \\
 & $l+0.1l'$ & \textbf{0.938} & \textbf{0.947} & \textbf{0.836} & \textbf{0.852} \\
\bottomrule
\end{tabular}
}
\end{table}

\subsection{Performance Comparison with SOTA} 
In this part, we compare our final model with the state-of-the-art (SOTA) models, \textit{i.e.}, BRISQUE~\cite{mittal2012no}, CORNIA~\cite{ye2012unsupervised}, HOSA~\cite{xu2016blind}, DeepBIQ~\cite{bianco2018use}, CNNIQA~\cite{kang2014convolutional}, DeepRN~\cite{varga2018deeprn}, and KonCept512~\cite{hosu2019koniq}.
The first three models map handcrafted features to image quality by SVR.
The fourth model maps the fine-tuned deep features to image quality by SVR.
The last three deep learning-based models adopt MAE, MSE, or their variant Huber loss for network training. 
And the results of these models are taken from Hosu et al.~\shortcite{hosu2019koniq}, while our results are obtained in the same setting.
From Table~\ref{tab:SOTA}, we can see that our model outperforms classic and current deep learning-based models. 
We note that by combining the loss and its variant with a weight of 1 and 0.1, our model can even achieve better results, \textit{e.g.}, SROCC values are 0.938 and 0.836 on KonIQ-10k test set and CLIVE, respectively.

\section{Conclusion}
\label{sec:conclusion}
Realizing that most IQA methods train models using MAE or MSE loss with empirically slow convergence, we address this problem by designing a class of loss functions with normalization.
The proposed loss includes Pearson correlation-induced loss as a special case.
And a special case of the loss variant is connected with RMSE between the linearly mapped predictions and the subjective ratings.
We theoretically prove that the embedded normalization helps to improve the smoothness of the loss landscape.
Besides, experimental verification of the proposed loss is conducted on the quality assessment of in-the-wild images.
Results on two benchmark datasets (KonIQ-10k and CLIVE) show that the model converges faster and better with the proposed loss than that with MAE loss and MSE loss.

The proposed loss is invariant to the scale of subjective ratings.
Facilitated with the new loss, we can easily mix multiple datasets with different scales of subjective ratings for training a universal IQA model.
In the future study, we intend to verify the effectiveness of this new loss in the universal image and video quality assessment problems.
Besides, it is a future direction on how to optimally choose the hyper-parameters $p$ and $q$ in the class of proposed loss functions for a specific regression task.

\begin{acks}
This work was partially supported by the National Natural Science Foundation of China under contracts 61572042, 61527804 and 61520106004. We also acknowledge High-Performance Computing Platform of Peking University for providing computational resources.
\end{acks}

\bibliographystyle{ACM-Reference-Format}
\bibliography{mmfp0612}

\appendix

\section{Derivation of $c$ in Eqn.~(6)}
\label{sec:supp_a}
\begin{lemma}
When $\mathbf{x}=(x_1,\cdots,x_N), 0< p_1\le p_2<+\infty$, we have the following norm inequality.
\begin{equation}
\|{\mathbf{x}}\|_{p_2}\le\|{\mathbf{x}}\|_{p_1}\le N^{\frac{1}{p_1}-\frac{1}{p_2}}\|{\mathbf{x}}\|_{p_2}
\end{equation}
\end{lemma}

\begin{proof}
We will separately prove the left part and the right part.

1. Proof of the left part.

Denote $y_i=|x_i|^{p_2}\ge0$. We have ${y_i}^{p_1/p_2}=|x_i|^{p_1}$, $0<p_1/p_2\le1$, and $0\le\frac{y_i}{\sum_{j=1}^N y_i}\le1$. Then
\begin{equation}
\sum_{i=1}^N \left(\frac{y_i}{\sum_{j=1}^N y_i}\right)^{p_1/p_2} \ge \sum_{i=1}^N {\frac{y_i}{\sum_{j=1}^N y_i}}=1
\end{equation}

That is
\begin{equation}
\sum_{i=1}^N {y_i}^{p_1/p_2} \ge \left(\sum_{i=1}^N {y_i}\right)^{p_1/p_2}
\end{equation}
\begin{equation}
\sum_{i=1}^N {|x_i|}^{p_1} \ge \left(\sum_{i=1}^N {|x_i|^{p_2}}\right)^{p_1/p_2}
\end{equation}
\begin{equation}
\left(\sum_{i=1}^N {|x_i|}^{p_1}\right)^{1/p_1} \ge \left(\sum_{i=1}^N {|x_i|^{p_2}}\right)^{1/p_2}
\end{equation}
\begin{equation}
\|{\mathbf{x}}\|_{p_2}\le\|{\mathbf{x}}\|_{p_1}
\end{equation}

2. Proof of the right part.

Based on H\"older inequality, we directly get
\begin{align}
\sum_{i=1}^N {|x_i|^{p_1}*1} \le & \left(\sum_{i=1}^N (|x_i|^{p_1})^{p_2/p_1}\right)^{p_1/p_2}\left(\sum_{i=1}^N {1^{1/(1-p_1/p_2)}}\right)^{1-p_1/p_2}\nonumber\\
=&N^{1-p_1/p_2}\left(\sum_{i=1}^N |x_i|^{p_2}\right)^{p_1/p_2}
\end{align}

Applying $p_1$-th root calculation to the above equation, we derive
\begin{equation}
\|{\mathbf{x}}\|_{p_1}\le N^{\frac{1}{p_1}-\frac{1}{p_2}}\|{\mathbf{x}}\|_{p_2}
\end{equation}

In summary, we proof the lemma.
\end{proof}

\textbf{Derivation of $c$ in Eqn.~(6)}:

Based on $p\ge1$ and the Minkowski inequality, we have
\begin{equation}
\|\hat{\mathbf{S}}-{\mathbf{S}}\|_p\le  \|\hat{\mathbf{S}}\|_p + \|{\mathbf{S}}\|_p
\end{equation}

Together with the above lemma and $\|\hat{\mathbf{S}}\|_q = \|{\mathbf{S}}\|_q=1$, we can derive
\begin{equation}
\|\hat{\mathbf{S}}\|_p + \|{\mathbf{S}}\|_p\le (\|\hat{\mathbf{S}}\|_q + \|{\mathbf{S}}\|_q)*
\begin{cases}
    N^{\frac{1}{p}-\frac{1}{q}} \\
    1 
\end{cases}
=
\begin{cases}
    2N^{\frac{1}{p}-\frac{1}{q}} & \mbox{if}\quad p<q, \\
    2 & \mbox{if}\quad p\ge q.
\end{cases}
\end{equation}

Thus
\begin{equation}
c = \max_{\|\hat{\mathbf{S}}\|_q = \|{\mathbf{S}}\|_q=1}{\|\hat{\mathbf{S}}-{\mathbf{S}}\|_p^p}=
\begin{cases}
    2^pN^{1-\frac{p}{q}} & \mbox{if}\quad p<q, \\
    2^p & \mbox{if}\quad p\ge q.
\end{cases}
\end{equation}

\section{$\hat{b}$ Curve in Our Experiment}
\label{sec:supp_b}
Figure~\ref{fig:bhat} shows the $\hat{b}$ curve with respect to iteration, and Figure~\ref{fig:avg bhat} shows the average $\hat{b}$ curve with respect to epoch. We can see that $\hat{b}$ is larger than 1 in practice.
\begin{figure}[!htb]
    \centering
    \includegraphics[width=.7\linewidth]{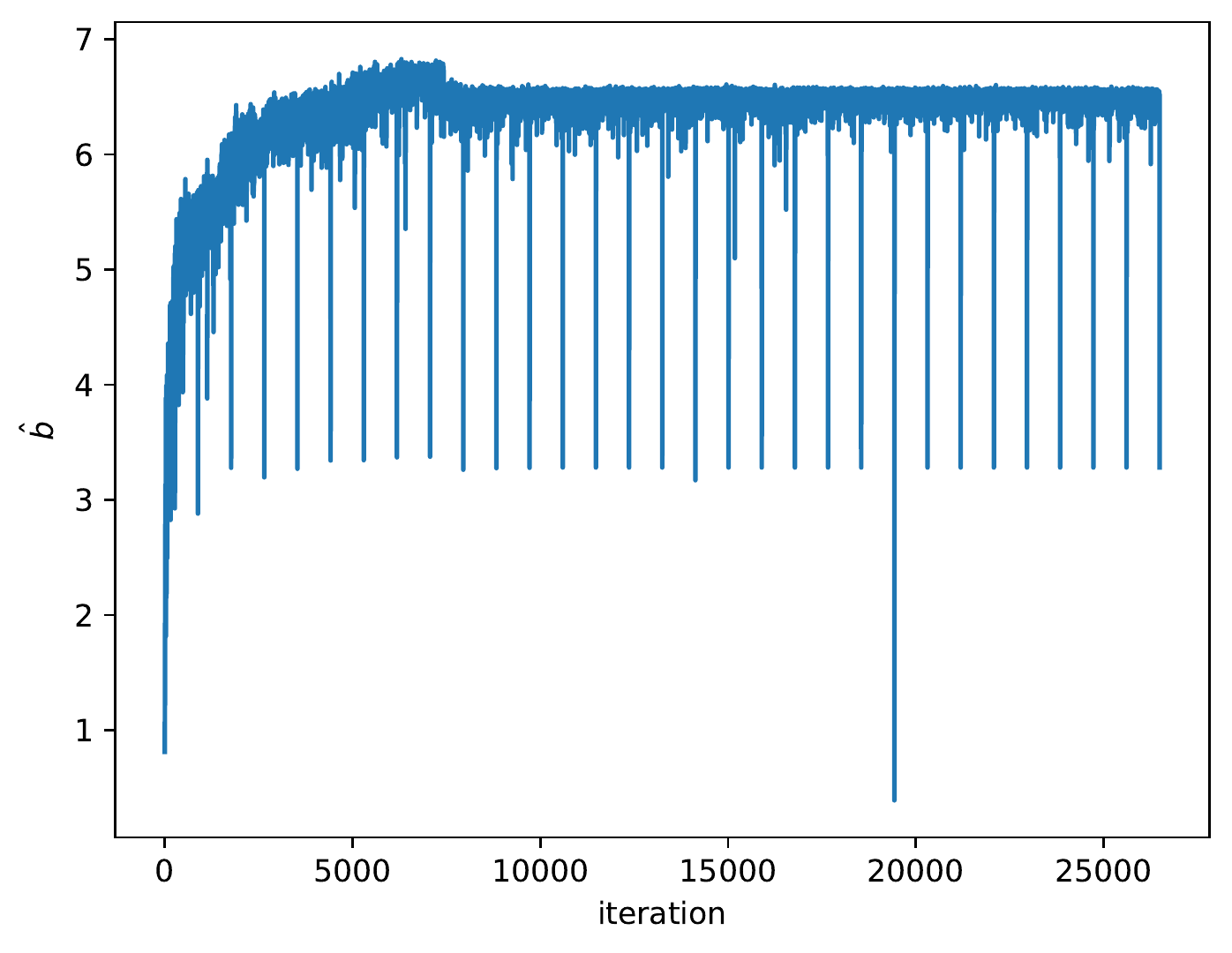}
    \caption{The $\hat{b}$ curve with respect to iteration in our experiment.}
    \label{fig:bhat}
\end{figure}

\begin{figure}[!htb]
    \centering
    \includegraphics[width=.7\linewidth]{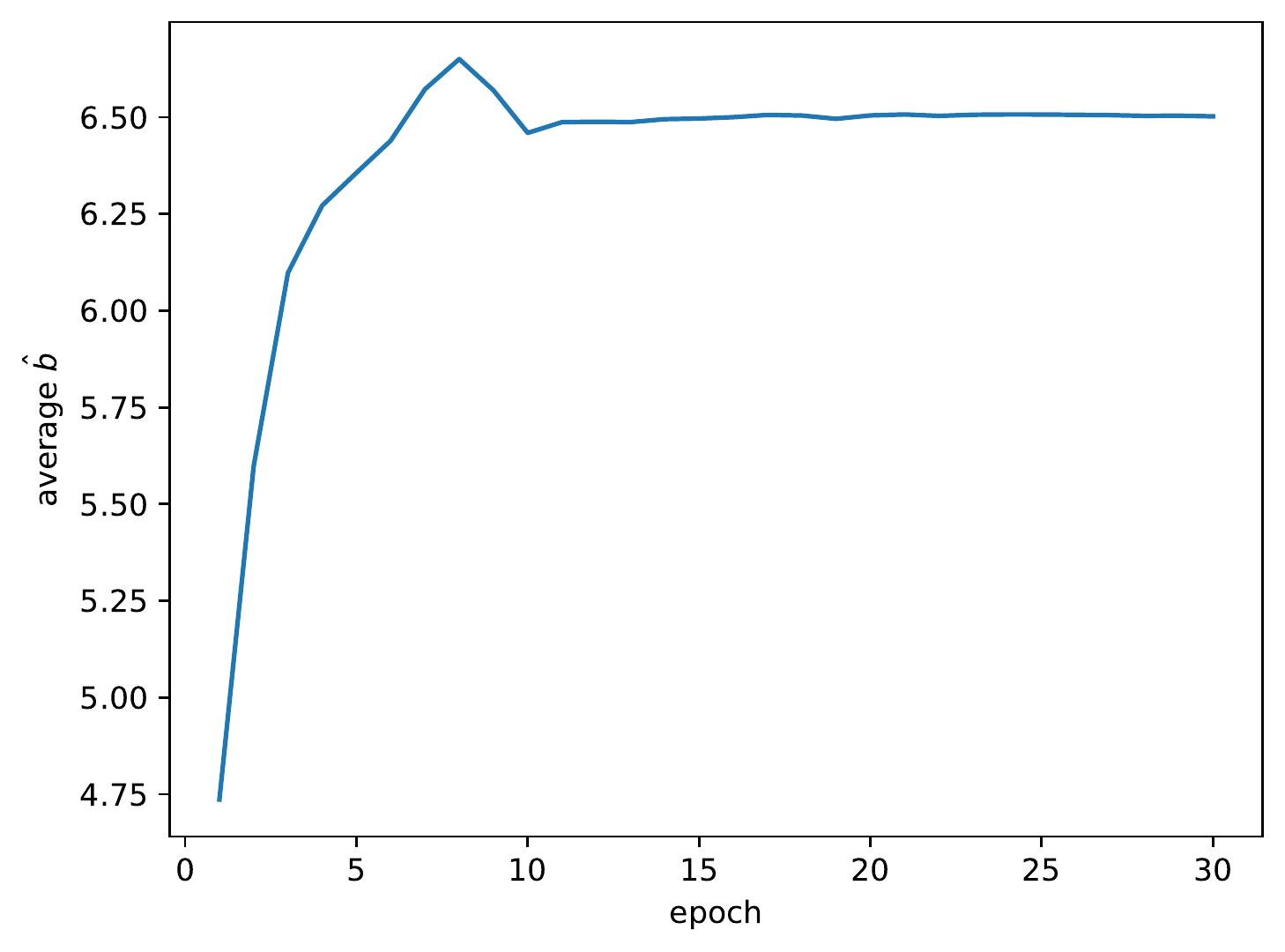}
    \caption{The average $\hat{b}$ curve with respect to epoch in our experiment.}
    \label{fig:avg bhat}
\end{figure}

\section{Proof of Theorem~\ref{theorem1} (Lipschitzness)}
\label{sec:supp_c}
\begin{proof}
Denote $\mathbf{Z}=\mathbf{g}-\frac{1}{N}\mathbf{1}^T\left\langle\mathbf{1}, \mathbf{g}\right\rangle$, and based on $\langle\mathbf{1}, \hat{\mathbf{S}}\rangle=0$, Eqn.~(19) becomes
\begin{align}
\mathbf{g}_n = \frac{1}{\hat{b}} \left\{\mathbf{Z}-\hat{\mathbf{S}}^T\langle\mathbf{Z},\hat{\mathbf{S}}\rangle\right\}
\end{align}

Thus
\begin{align}
\left\|\mathbf{g}_n\right\|^2 =& \frac{1}{\hat{b}^2} \left\|\mathbf{Z}-\hat{\mathbf{S}}^T\langle\mathbf{Z},\hat{\mathbf{S}}\rangle\right\|^2\nonumber\\
=& \frac{1}{\hat{b}^2} \left\{\|\mathbf{Z}\|^2-2\langle\mathbf{Z},\hat{\mathbf{S}}\rangle^2+\langle\hat{\mathbf{S}},\hat{\mathbf{S}}\rangle\langle\mathbf{Z},\hat{\mathbf{S}}\rangle^2\right\}\nonumber\\
=& \frac{1}{\hat{b}^2} \left\{\|\mathbf{Z}\|^2-2\langle\mathbf{Z},\hat{\mathbf{S}}\rangle^2+\langle\mathbf{Z},\hat{\mathbf{S}}\rangle^2\right\} (\mbox{Because}\ \|\hat{\mathbf{S}}\|^2=1) \nonumber\\
=& \frac{1}{\hat{b}^2} \left\{\|\mathbf{Z}\|^2-\langle\mathbf{Z},\hat{\mathbf{S}}\rangle^2\right\} \nonumber\\
=& \frac{1}{\hat{b}^2} \left\{\left\|\mathbf{g}-\frac{1}{N}\mathbf{1}^T\left\langle\mathbf{1}, \mathbf{g}\right\rangle\right\|^2-\left\langle\mathbf{g}-\frac{1}{N}\mathbf{1}^T\left\langle\mathbf{1}, \mathbf{g}\right\rangle,\hat{\mathbf{S}}\right\rangle^2\right\} \nonumber\\
=& \frac{1}{\hat{b}^2} \left\{\left\|\mathbf{g}-\frac{1}{N}\mathbf{1}^T\left\langle\mathbf{1}, \mathbf{g}\right\rangle\right\|^2-\left(\left\langle\mathbf{g},\hat{\mathbf{S}}\right\rangle-\frac{1}{N}\left\langle\mathbf{1}, \hat{\mathbf{S}}\right\rangle\left\langle\mathbf{1}, \mathbf{g}\right\rangle\right)^2\right\} \nonumber\\
=& \frac{1}{\hat{b}^2} \left\{\left\|\mathbf{g}-\frac{1}{N}\mathbf{1}^T\left\langle\mathbf{1}, \mathbf{g}\right\rangle\right\|^2-\left\langle\mathbf{g},\hat{\mathbf{S}}\right\rangle^2\right\} \nonumber\\
=& \frac{1}{\hat{b}^2} \left\{\left\|\mathbf{g}\right\|^2-\frac{1}{N}\left\langle\mathbf{1}, \mathbf{g}\right\rangle^2-\left\langle\mathbf{g},\hat{\mathbf{S}}\right\rangle^2\right\} 
\end{align}
\end{proof}

\section{Proof of Theorem~\ref{theorem2} ($\beta$-smoothness)}
\label{sec:supp_d}
\begin{proof}
When $q$ equals 2, Eqn.~(7) becomes
\begin{equation}
    \frac{\partial \hat{\mathbf{S}}}{\partial \hat{\mathbf{Q}}} = \frac{1}{\hat{b}}\left\{\mathbf{I}-\frac{1}{N}\mathbf{1}\mathbf{1}^T-\hat{\mathbf{S}}\hat{\mathbf{S}}^T\right\}=\frac{1}{\hat{b}}\mathbf{K}
\end{equation}

And based on $\langle\mathbf{1}, \hat{\mathbf{S}}\rangle=0, \langle\hat{\mathbf{S}}, \hat{\mathbf{S}}\rangle=1$, we derive the property of $\mathbf{K}$ as follow.
\begin{equation}
    \mathbf{K}\hat{\mathbf{S}}=\mathbf{0},\hat{\mathbf{S}}^T\mathbf{K}=\mathbf{0}^T, \mathbf{K}^2=\mathbf{K}=\mathbf{K}^T
\end{equation}

Denote $\mathbf{M}= \mathbf{1}\mathbf{1}^T=\mathbf{M}^T$, and we know $\mathbf{H}=\mathbf{H}^T$. Thus
\begin{equation}
    \mathbf{g}_n=\frac{1}{\hat{b}}\mathbf{K}{\mathbf{g}}
\end{equation}
\begin{equation}
    \frac{\partial \mathbf{1}\left\langle\mathbf{1}, {\mathbf{g}}\right\rangle}{\partial \hat{\mathbf{S}}} =  \mathbf{1}\mathbf{1}^T\frac{\partial^2 l}{\partial \hat{\mathbf{S}}^2}=\mathbf{M}\mathbf{H}
\end{equation}
\begin{align}
    \frac{\partial \hat{\mathbf{S}}\langle{\mathbf{g}},\hat{\mathbf{S}}\rangle}{\partial \hat{\mathbf{S}}} =& \langle{\mathbf{g}},\hat{\mathbf{S}}\rangle\frac{\partial \hat{\mathbf{S}}}{\partial \hat{\mathbf{S}}} + \hat{\mathbf{S}}\frac{\partial \langle{\mathbf{g}},\hat{\mathbf{S}}\rangle}{\partial \hat{\mathbf{S}}}\nonumber\\
    =& \langle{\mathbf{g}},\hat{\mathbf{S}}\rangle + \hat{\mathbf{S}}\hat{\mathbf{S}}^T\mathbf{H}+\hat{\mathbf{S}}{\mathbf{g}}^T
\end{align}
\begin{align}
    \frac{\partial \frac{1}{\hat{b}}}{\partial \hat{\mathbf{Q}}} = -\frac{1}{\hat{b}^3}(\hat{\mathbf{Q}}^T-\mathbf{1}^T\hat{a})=-\frac{1}{\hat{b}^2}\hat{\mathbf{S}}^T
\end{align}

So
\begin{align}
    \mathbf{H}_n = & \frac{\partial }{\partial \hat{\mathbf{Q}}}\left\{\frac{1}{\hat{b}} \left[ {\mathbf{g}}-\frac{1}{N}\mathbf{1}\left\langle\mathbf{1}, {\mathbf{g}}\right\rangle-\hat{\mathbf{S}}\langle{\mathbf{g}},\hat{\mathbf{S}}\rangle\right]\right\} \nonumber\\
    = & \frac{1}{\hat{b}}\frac{\partial }{\partial \hat{\mathbf{S}}}\left\{ {\mathbf{g}}-\frac{1}{N}\mathbf{1}\left\langle\mathbf{1}, {\mathbf{g}}\right\rangle-\hat{\mathbf{S}}\langle{\mathbf{g}},\hat{\mathbf{S}}\rangle\right\}\frac{\partial \hat{\mathbf{S}}}{\partial \hat{\mathbf{Q}}} + \left\{{\mathbf{g}}-\frac{1}{N}\mathbf{1}\left\langle\mathbf{1}, {\mathbf{g}}\right\rangle-\hat{\mathbf{S}}\langle{\mathbf{g}},\hat{\mathbf{S}}\rangle\right\}\frac{\partial \frac{1}{\hat{b}}}{\partial \hat{\mathbf{S}}} \nonumber\\
    = & \frac{1}{\hat{b}^2}\left\{ \mathbf{H}-\frac{1}{N}\mathbf{M}\mathbf{H}-\langle{\mathbf{g}},\hat{\mathbf{S}}\rangle-\hat{\mathbf{S}}\hat{\mathbf{S}}^T\mathbf{H}-\hat{\mathbf{S}}{\mathbf{g}}^T\right\}\left\{\mathbf{I}-\frac{1}{N}\mathbf{M}-\hat{\mathbf{S}}\hat{\mathbf{S}}^T\right\} \nonumber\\
    & - \frac{1}{\hat{b}^2}\left\{{\mathbf{g}}-\frac{1}{N}\mathbf{1}\left\langle\mathbf{1}, {\mathbf{g}}\right\rangle-\hat{\mathbf{S}}\langle{\mathbf{g}},\hat{\mathbf{S}}\rangle\right\}\hat{\mathbf{S}}^T \nonumber\\  
    = & \frac{1}{\hat{b}^2} \left\{\mathbf{I}-\frac{1}{N}\mathbf{M}-\hat{\mathbf{S}}\hat{\mathbf{S}}^T\right\}\mathbf{H}\left\{\mathbf{I}-\frac{1}{N}\mathbf{M}-\hat{\mathbf{S}}\hat{\mathbf{S}}^T\right\}\nonumber\\
    & -\frac{1}{\hat{b}^2}\left\{\langle{\mathbf{g}},\hat{\mathbf{S}}\rangle+\hat{\mathbf{S}}{\mathbf{g}}^T\right\}\left\{\mathbf{I}-\frac{1}{N}\mathbf{M}-\hat{\mathbf{S}}\hat{\mathbf{S}}^T\right\} \nonumber\\
    & - \frac{1}{\hat{b}^2}\left\{\mathbf{I}-\frac{1}{N}\mathbf{M}-\hat{\mathbf{S}}\hat{\mathbf{S}}^T\right\}{\mathbf{g}}\hat{\mathbf{S}}^T 
 \end{align}
 
 Note that $\mathbf{K}=\mathbf{I}-\frac{1}{N}\mathbf{1}\mathbf{1}^T-\hat{\mathbf{S}}\hat{\mathbf{S}}^T$, and we have 
 \begin{align}
    \mathbf{H}_n = \frac{1}{\hat{b}^2} \left(\mathbf{K}\mathbf{H}\mathbf{K}-\langle{\mathbf{g}},\hat{\mathbf{S}}\rangle\mathbf{K}-\hat{\mathbf{S}}{\mathbf{g}}^T\mathbf{K}-\mathbf{K}{\mathbf{g}}\hat{\mathbf{S}}^T\right)
\end{align}

Then
\begin{align}
    \hat{b}^4 \mathbf{g}_n^T\mathbf{H}_n\mathbf{g}_n = &  {\mathbf{g}}^T\mathbf{K}\left(\mathbf{K}\mathbf{H}\mathbf{K}-\langle{\mathbf{g}},\hat{\mathbf{S}}\rangle\mathbf{K}-\hat{\mathbf{S}}{\mathbf{g}}^T\mathbf{K}-\mathbf{K}{\mathbf{g}}\hat{\mathbf{S}}^T\right)\mathbf{K}{\mathbf{g}} \nonumber\\
    = & {\mathbf{g}}^T\mathbf{K}^2\mathbf{H}\mathbf{K}^2{\mathbf{g}}-\langle{\mathbf{g}},\hat{\mathbf{S}}\rangle{\mathbf{g}}^T\mathbf{K}^3{\mathbf{g}}-{\mathbf{g}}^T\mathbf{K}\hat{\mathbf{S}}{\mathbf{g}}^T\mathbf{K}^2{\mathbf{g}}-{\mathbf{g}}^T\mathbf{K}^2{\mathbf{g}}\hat{\mathbf{S}}^T\mathbf{K}{\mathbf{g}} \nonumber\\
    = & {\mathbf{g}}^T\mathbf{K}\mathbf{H}\mathbf{K}{\mathbf{g}}-\langle{\mathbf{g}},\hat{\mathbf{S}}\rangle{\mathbf{g}}^T\mathbf{K}^2{\mathbf{g}}-{\mathbf{g}}^T(\mathbf{K}\hat{\mathbf{S}}){\mathbf{g}}^T\mathbf{K}{\mathbf{g}}-{\mathbf{g}}^T\mathbf{K}{\mathbf{g}}(\hat{\mathbf{S}}^T\mathbf{K}){\mathbf{g}}\nonumber\\
    = & ({\mathbf{g}}^T\mathbf{K})\mathbf{H}(\mathbf{K}{\mathbf{g}})-\langle{\mathbf{g}},\hat{\mathbf{S}}\rangle({\mathbf{g}}^T\mathbf{K})(\mathbf{K}{\mathbf{g}})\nonumber\\
    = & \hat{b}^2 \left\{\mathbf{g}_n^T\mathbf{H}\mathbf{g}_n-\langle{\mathbf{g}},\hat{\mathbf{S}}\rangle\left\|\mathbf{g}_n\right\|^2\right\}
\end{align}

That is
\begin{align}\label{eq:beta-smoothness}
    \mathbf{g}_n^T\mathbf{H}_n\mathbf{g}_n = \frac{1}{\hat{b}^2} \left\{\mathbf{g}_n^T\mathbf{H}\mathbf{g}_n-\langle{\mathbf{g}},\hat{\mathbf{S}}\rangle\left\|\mathbf{g}_n\right\|^2\right\}
\end{align}

Further, when $p$ equals 2, we derive ${\mathbf{g}} =\frac{2}{c}(\hat{\mathbf{S}}-{\mathbf{S}}), \mathbf{H}=\frac{2}{c}\mathbf{I}$. Then
\begin{align}
\mathbf{g}_n^T\mathbf{H}\mathbf{g}_n = & \left(\frac{1}{\hat{b}}\mathbf{K}{\mathbf{g}}\right)^T\frac{2}{c}\mathbf{I}\left(\frac{1}{\hat{b}}\mathbf{K}{\mathbf{g}}\right)\nonumber\\
= & \frac{2}{c\hat{b}^2}{\mathbf{g}}^T\mathbf{K}^T\mathbf{K}{\mathbf{g}} = \frac{2}{c\hat{b}^2}{\mathbf{g}}^T\mathbf{K}^2{\mathbf{g}} = \frac{2}{c\hat{b}^2}{\mathbf{g}}^T\mathbf{K}{\mathbf{g}} \nonumber\\
= & \frac{2}{c\hat{b}^2}{\mathbf{g}}^T\left\{\mathbf{I}-\frac{1}{N}\mathbf{1}\mathbf{1}^T-\hat{\mathbf{S}}\hat{\mathbf{S}}^T\right\}{\mathbf{g}} \nonumber\\
= & \frac{1}{\hat{b}^2}\left\{{\mathbf{g}}^TH{\mathbf{g}}-\frac{2}{cN}\langle{\mathbf{g}},\mathbf{1}\rangle^2-\frac{2}{c}\langle{\mathbf{g}},\hat{\mathbf{S}}\rangle^2\right\}
\end{align}

\begin{align}
\langle{\mathbf{g}},\hat{\mathbf{S}}\rangle = \langle\frac{2}{c}(\hat{\mathbf{S}}-\mathbf{S}),\hat{\mathbf{S}}\rangle=\frac{2}{c}(1-\langle\mathbf{S},\hat{\mathbf{S}}\rangle)
\end{align}

And
\begin{align}
\left\langle{\mathbf{g}},\mathbf{1}\right\rangle=\left\langle\frac{2}{c}(\hat{\mathbf{S}}-{\mathbf{S}}),\mathbf{1}\right\rangle=\frac{2}{c}\left(\left\langle\hat{\mathbf{S}},\mathbf{1}\right\rangle-\left\langle{\mathbf{S}},\mathbf{1}\right\rangle\right)=0
\end{align}

Based on the above equation and Theorem~\ref{theorem1}, we have
\begin{align}
\left\|\mathbf{g}_n\right\|^2 = \frac{1}{\hat{b}^2} \left\{\left\|{\mathbf{g}}\right\|^2-\left\langle{\mathbf{g}},\hat{\mathbf{S}}\right\rangle^2\right\} 
\end{align}

Use the above four equations to substitute the corresponding parts in Eqn.~(\ref{eq:beta-smoothness}), we can derive the following equation.
\begin{align}
    {\mathbf{g}}_n^T\mathbf{H}_n{\mathbf{g}}_n = & \frac{1}{\hat{b}^4} \left\{\mathbf{g}^T\mathbf{H}\mathbf{g}\right. \nonumber\\
    & \left.-\frac{2}{c}\left(1-\langle{\mathbf{S}},\hat{\mathbf{S}}\rangle\right)\left[\frac{4}{c^2}\left(1-\langle{\mathbf{S}},\hat{\mathbf{S}}\rangle\right)+\left\|{\mathbf{g}}\right\|^2-\left\langle{\mathbf{g}},\hat{\mathbf{S}}\right\rangle^2\right]\right\}.
\end{align}

\end{proof}

\textit{Remark}: Based on Cauchy inequality, $\langle{\mathbf{S}},\hat{\mathbf{S}}\rangle\le\sqrt{\|\mathbf{S}\|^2\|\hat{\mathbf{S}}\|^2}=1, \langle{\mathbf{g}},\hat{\mathbf{S}}\rangle^2\le\|{\mathbf{g}}\|^2\|\hat{\mathbf{S}}\|^2=\|{\mathbf{g}}\|^2$. So $\left(1-\langle{\mathbf{S}},\hat{\mathbf{S}}\rangle\right)\left\{\frac{4}{c^2}\left(1-\langle{\mathbf{S}},\hat{\mathbf{S}}\rangle\right)+\left\|{\mathbf{g}}\right\|^2-\left\langle{\mathbf{g}},\hat{\mathbf{S}}\right\rangle^2\right\}\ge0$. 
This indicates that the embedded normalization reduces the quadratic form of the loss Hessian. 
At the meantime, since $\left\langle{\mathbf{g}},\mathbf{1}\right\rangle=0$, $\frac{\partial l}{\partial \hat{a}}$ does not contribute to the reduction and it is all provided by $\frac{\partial l}{\partial \hat{b}}$.

\section{Additional Results}

\subsection{Different Optimizers}
Besides Adam, we show additional results with the SGD/Adadelta optimizer.
The experimental setting for Table~\ref{tab:optimizer} is similar to the experimental setting for Figure~\ref{fig:sensitivity}\subref{fig:learningrate}. We just replaced Adam with SGD or Adadelta. It can be seen that the model performances are very sensitive to the initial learning rates when using the SGD or Adadelta optimizer. Besides, the best initial learning rate is 1e-1 for the proposed loss, 1e-2 for the MAE loss, and 1e-4 for the MSE loss.
From Table~\ref{tab:optimizer}, we can see that the proposed loss is better than the MAE and MSE losses when the SGD or Adadelta optimizer is used and the best initial learning rate is chosen. 

\begin{table}[!htb]
\centering
\caption{PLCC comparisons under SGD/Adadelta optimizer}
\label{tab:optimizer}
\begin{tabular}{lcccc}  
\toprule
Initial learning rate & MAE loss & MSE loss & Proposed loss \\
\midrule
1e-1 & 0.843/0.780 & Failed/0.069 & \textbf{0.931}/\textbf{0.930} \\
1e-2 & 0.909/0.861 & 0.781/0.690 & 0.916/0.911 \\
1e-3 & 0.868/0.068 & 0.839/0.701 & 0.899/0.889 \\
1e-4 & 0.620/0.007 & 0.890/0.739 & 0.868/0.808 \\
1e-5 & 0.090/0.138 & 0.851/0.458 & 0.770/0.458 \\
\bottomrule
\end{tabular}
\end{table}

\subsection{Different Architectures}
\label{sec:supp_e}
The experimental setting for Table~\ref{tab:architecture} is similar to the experimental setting for Sec.~\ref{sec:4.3.4}. We just used the non-BN network architecture (AlexNet or VGG-16) instead of the ResNet-based network as the backbone. 
From Table~\ref{tab:architecture}, we can see that the proposed loss is better than the MSE and MSE losses when AlexNet/VGG-16 is used as the backbone. 
Together with the experiments on ResNet-based backbones, it can be seen that Batch Normalization layers in networks will not affect the superiority new loss.

\begin{table}[!htb]
\centering
\caption{SROCC/PLCC comparisons under non-BN network architectures}
\label{tab:architecture}
\begin{tabular}{lcccc}  
\toprule
Backbone & MAE loss & MSE loss & Proposed loss \\
\midrule
AlexNet & 0.788/0.779 & 0.811/0.799 & 0.879/0.886 \\
VGG-16 & 0.840/0.834 & 0.844/0.835 & 0.910/0.913 \\
\bottomrule
\end{tabular}
\end{table}

\subsection{t-test}
In the paper, the results shown in Table~\ref{tab:pq} and Table~\ref{tab:SOTA} were based on the experiments on one train-validation-test split provided by the KonIQ-10k dataset's owner. However, for performing the t-test to verify whether the performance gains in Table~\ref{tab:pq} and Table~\ref{tab:SOTA} are statistically significant or not, we need several experiments on different train-validation-test splits.
We conducted experiments on 10 random splits of CLIVE with the ResNet-18 backbone and performed the t-test for different combinations of p and q, as well as ``$l$'' and ``$l+0.1l'$''. 
The results show that, in terms of PLCC, 
(a) ``$p=1, q=2$'' is significantly better than ``$p=2, q=1$'' (p-value 0.018) and ``$p=2, q=2$'' (p-value 0.028), while it is on par with ``$p=1, q=1$'' (p-value 0.385). 
(b) ``$l+0.1l'$'' is slightly (but not significantly) better than ``$l$'' (p-value 0.172).


\end{document}